\documentclass{new_tlp}
\usepackage{latexsym}

\usepackage{etex}

\usepackage{latexsym}
\usepackage{amssymb}
\usepackage{amsfonts}
\usepackage{amsmath}
\usepackage{verbatim}
\usepackage{url}
\usepackage[usenames]{color}
\usepackage{stmaryrd}
\usepackage{xy}
\usepackage{xspace}
\usepackage{tikz}
\usetikzlibrary{shapes}
\usepackage[noend]{algpseudocode}
\usepackage{hyperref}

\usepackage{tcolorbox}
\usepackage{multirow}
\usetikzlibrary{automata}
\usetikzlibrary{arrows}
\usetikzlibrary{backgrounds}
\usetikzlibrary{positioning}
\usetikzlibrary{fit}
\usetikzlibrary{calc}
\usetikzlibrary{matrix}
\usetikzlibrary{patterns}
\hypersetup{
  linkcolor  = blue,
  citecolor  = blue,
  urlcolor   = blue,
  colorlinks = true,
}

\newtheorem{cla}{Claim}
\newtheorem{theorem}{Theorem}
\newtheorem{lemma}{Lemma}
\newtheorem{definition}{Definition}
\newtheorem{observation}{Observation}
\usepackage{todonotes}
\usepackage{stackengine}
\usepackage{xspace}

\newcommand{\MLSFour}{\mbox{\rm S4}}
\newcommand{\MLSFive}{\mbox{\rm S5}}
\newcommand{\MLT}{\mbox{\rm T}}
\newcommand{\MLTB}{\mbox{\rm TB}}
\newcommand{\MLDFour}{\mbox{\rm D4}}
\newcommand{\MLDFourFive}{\mbox{\rm D45}}
\newcommand{\MLD}{\mbox{\rm D}}
\newcommand{\MLDFive}{\mbox{\rm D5}}

\newcommand{\MLKFour}{\mbox{\rm K4}}
\newcommand{\MLKFourFive}{\mbox{\rm K45}}
\newcommand{\MLKBFourFive}{\mbox{\rm KB45}}
\newcommand{\MLKFive}{\mbox{\rm K5}}
\newcommand{\MLK}{\mbox{\rm K}}

\newcommand{\ML}[1]{\mbox{\rmfamily\bfseries\scshape #1}\xspace}

\newcommand{\CTwo}{$\mathrm{C}^2$\xspace}
\newcommand{\COne}{$\mathrm{C}^1$\xspace}
\newcommand{\GCTwo}{$\mathrm{GC}^2$\xspace}

\newcommand{\GCTwoo}{\mathrm{GC}^2}

\newcommand{\RevDiamond}{\ensurestackMath{\stackengine{.5pt}{\Diamond}{\scalebox{.75}[1]{$-$}}{O}{c}{F}{F}{L}}}

\newcommand{\RevBox}{\ensurestackMath{\stackengine{.5pt}{\Box}{\scalebox{.75}[1]{$-$}}{O}{c}{F}{F}{L}}}

\newcommand{\U}{\mathbf{U}}

\newcommand{\LogSpace}{\textsc{LogSpace}}

\newcommand{\NP}{\textsc{NP}}

\newcommand{\PSpace}{\textsc{PSpace}}
\newcommand{\ExpTime}{\textsc{ExpTime}}

\newcommand{\NExpTime}{\textsc{NExpTime}}
\newcommand{\TwoExpTime}{2\textsc{-ExpTime}}

\newcommand{\deff}{\stackrel{\text{def}}{=}}

\newcommand{\str}[1]{{\mathfrak{#1}}}
\newcommand{\tuple}[1]{\langle#1\rangle} 
\newcommand{\N}{{\mathbb N}}   
\newcommand{\restr}{\!\!\restriction\!\!}
\renewcommand{\iff}{\leftrightarrow}

\newcommand{\st}{\mathbf{st}} 
\newcommand{\stx}{\st_x} 
\newcommand{\sty}{\st_y} 
\newcommand{\stv}{\st_v} 

\newcommand{\tr}{\mathbf{tr}} 

\newcommand{\phiPartition}{\varphi_{\text{partition}}}
\newcommand{\phiTorus}{\varphi_{\text{torus}}}
\newcommand{\phiTorusSize}{\varphi_{\text{torusSize}}}
\newcommand{\phiChessboard}{\varphi_{\text{chessboard}}}
\newcommand{\phiFirstCell}{\varphi_{\text{firstCell}}}
\newcommand{\phiSuccessors}{\varphi_{\text{succ}}}

\newcommand{\lan}{\mathit{lantern}}
\newcommand{\inn}{\mathit{inner}}
\newcommand{\bl}{\mathit{blk}}
\newcommand{\wht}{\mathit{wht}}

\newcommand{\vbw}{\mathit{vbw}}
\newcommand{\hbw}{\mathit{hbw}}
\newcommand{\hwb}{\mathit{hwb}}
\newcommand{\vwb}{\mathit{vwb}}

\newcommand{\phivbw}{\varphi_\mathit{vbw}}
\newcommand{\phihbw}{\varphi_\mathit{hbw}}
\newcommand{\phihwb}{\varphi_\mathit{hwb}}
\newcommand{\phivwb}{\varphi_\mathit{vwb}}

\newcommand{\phiPseudoUniqueness}{\varphi_{\text{pseudoUniqueness}}}
\newcommand{\phiEqualH}{\varphi_{\text{equalH}}}
\newcommand{\phiEqualV}{\varphi_{\text{equalV}}}
\newcommand{\phiAddOneVBW}{\varphi_{V_w {=} V_b {\oplus_{2^n}} 1 } }
\newcommand{\phiAddOneHBW}{\varphi_{H_w {=} H_b {\oplus_{2^n}} 1 } }
\newcommand{\phiAddOneHWB}{\varphi_{H_b {=} H_w {\oplus_{2^n}} 1 } }
\newcommand{\phiAddOneVWB}{\varphi_{V_b {=} V_w {\oplus_{2^n}} 1 } }

\newcommand{\phiTiling}{\varphi_{\text{tiling}}}
\newcommand{\phiReduction}{\varphi_{\text{reduction}}}
\newcommand{\phiTile}{\varphi_{\text{tile}}}
\newcommand{\phiInitCond}{\varphi_{\text{initCond}}}
\newcommand{\phiTilingRules}{\varphi_{\text{tilingRules}}}

\let\cH\undefined
\let\cA\undefined
\let\cB\undefined
\let\cT\undefined
\let\cV\undefined
\let\cP\undefined
\let\cF\undefined
\let\cL\undefined

\newcommand{\cH}{\mathcal{H}}
\newcommand{\cA}{\mathcal{A}}
\newcommand{\cB}{\mathcal{B}}
\newcommand{\cT}{\mathcal{T}}
\newcommand{\cV}{\mathcal{V}}
\newcommand{\cP}{\mathcal{P}}
\newcommand{\cL}{\mathcal{L}}
\newcommand{\cF}{\mathcal{F}}

\newcommand{\irref}{\text{\emph{irref}}}

\title[Complexity of Graded Modal Logics with Converse]
{Completing the Picture: Complexity of Graded Modal Logics with Converse}

\author[B. Bednarczyk, E. Kiero\'nski and P. Witkowski]
	{Bartosz Bednarczyk\\
	Institute of Computer Science, University of Wroc{\l}aw, Poland and\\
	Computational Logic Group, TU Dresden, Germany\\
	\email{bartosz.bednarczyk@cs.uni.wroc.pl}
	\and 
	Emanuel Kiero\'nski\\
	Institute of Computer Science, University of Wroc{\l}aw, Poland\\
	\email{emanuel.kieronski@cs.uni.wroc.pl}
	\and
	Piotr Witkowski\\
	Institute of Computer Science, University of Wroc{\l}aw, Poland\\
	\email{piotr.witkowski@cs.uni.wroc.pl}
	}

\jdate{???}
\pubyear{???}
\pagerange{\pageref{firstpage}--\pageref{lastpage}}
\doi{???}

\begin{document} \label{firstpage}
\maketitle			


\begin{abstract}
A complete classification of the complexity of the local and global 
satisfiability problems for graded modal language over traditional classes of 
frames have already been established. By ''traditional'' classes of frames, we 
mean those characterized by any positive combination of reflexivity, seriality, 
symmetry, transitivity, and the Euclidean property. 
In this paper, we fill the gaps remaining in an analogous classification of the 
graded modal language with graded converse modalities. In particular, we show 
its~$\NExpTime$-completeness over the class of Euclidean frames, demonstrating 
this way that over this class the considered language is harder than the 
language without graded modalities or without converse modalities. We also 
consider its variation disallowing graded converse modalities, but still 
admitting basic converse modalities. Our most important result for this 
variation is confirming an earlier conjecture that it is decidable over 
transitive frames. This contrasts with the undecidability of the language 
with graded converse modalities.

Under consideration in Theory and Practice of Logic Programming (TPLP).
\end{abstract}



\section{Introduction}

For many years modal logic has been  an active topic in many academic disciplines, including philosophy, mathematics, linguistics,
and computer science. Regarding applications in computer science, e.g., in knowledge representation or verification, some important variations are those involving graded and converse modalities. In this paper, we investigate their 
computational complexity.

 By  \emph{a modal logic} we will mean a pair~$(\cL, \cF)$, represented usually
as~$\cF(\cL^*)$, where $\cL$ is a \emph{modal language},~$\cF$ is a \emph{class of frames}, and~$\cL^*$ is a short
symbolic representation of~$\cL$ (see the next paragraph), characterizing the \emph{modalities} of~$\cL$.
For example~$\MLKFour(\Diamond_{\ge})$ will denote the graded modal logic of transitive frames.

While we are mostly interested in languages with graded and converse
modalities, to set the scene we need to mention languages without
them. Overall, the following five languages are relevant: the basic
\emph{one-way modal language} ($\cL^*= \Diamond$) containing only one,
\emph{forward}, modality~$\Diamond$; \emph{graded one-way modal
  language} ($\cL^* = \Diamond_{\ge}$) extending the previous one by
\emph{graded} forward modalities,~$\Diamond_{\ge n}$, for all~$n \in \N$; \emph{two-way modal language} ($\cL^* = \Diamond, \RevDiamond$) containing basic forward modality
and the \emph{converse} modality~$\RevDiamond$; \emph{graded two-way
  modal language} ($\cL^* = \Diamond_{\ge}, \RevDiamond_{\ge}$)
containing the forward modality, the converse modality and their
graded versions~$\Diamond_{\ge n}$,~$\RevDiamond_{\ge n}$, for all~$n \in N$; and, additionally, a restriction of the latter without
graded converse modalities, but with basic converse modality ($\cL^* =\Diamond_{\ge}, \RevDiamond$).

The meaning of graded modalities is natural:~$\Diamond_{\ge n} \varphi$ means ''$\varphi$ is true at no fewer than~$n$ successors of the current world'', and~$\RevDiamond_{\ge} \varphi$ means ''$\varphi$ is true at no fewer than~$n$ predecessors of the current world''. We also recall that $\Diamond \varphi$ means
''$\varphi$ is true at some successor of the current world'' and $\RevDiamond \varphi$ --- ''$\varphi$ is true at some predecessor of the current world''. Thus, e.g.,~$\Diamond$ is equivalent to~$\Diamond_{\ge 1}$. 

Our aim is to classify the complexity of the local (``in a world'')
and global (``in all worlds'') satisfiability problems for all the
logics obtained by combining any of the above languages with any class
of frames from the so-called modal cube, that is a class of frames
characterized by any positive combination of the axioms of reflexivity~(T),
seriality~(D), symmetry~(B), transitivity~(4), and 
the Euclidean property~(5).

See Fig.~\ref{fig:figure1} for a visualization of the
modal cube. Nodes of the depicted graph correspond to classes of
frames and are labelled by letters denoting the above-mentioned
properties, with S used in S4 and S5 for some historical reasons to
denote reflexivity, and K denoting the class of all frames. 
If there is a path from a class~$X$ to a class~$Y$ then it means
that any class from~$Y$ also belongs to~$X$ (as all the axioms of~$X$ are also present in~$Y$). 
Note that the modal cube contains only~$15$ classes, since some different
combinations of the relevant axioms lead to identical classes,
e.g., reflexivity implies seriality, symmetry and transitivity imply Euclideanness, and so on.
\begin{figure}
\includegraphics{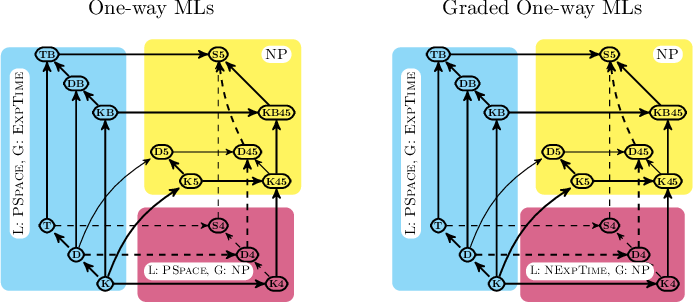}
\caption{Complexity of one-way modal logics. All bounds are tight. If local and global satisfiability differ
in complexity then ''L:'' indicates local and ''G:''---global satisfiability.} 
\label{fig:figure1}
\end{figure}

A lot of work has been already done.  The cases of basic one-way language and graded one-way language
are completely understood. See Fig.~\ref{fig:figure1}.
The results for the former can be established using some standard techniques,
see, e.g.,~\cite{BBV01} and the classical paper~\cite{Lad77}. The local satisfiability
of the latter is systematically analysed in~\cite{KazakovP09}, with complexities turning out to lie between \NP{}
and \NExpTime{}. As for its global satisfiability, some of the results follow from~\cite{KazakovP09}, some are given
in~\cite{Zolin17}, and the other can be easily obtained using again some standard techniques.
\begin{figure}
\includegraphics{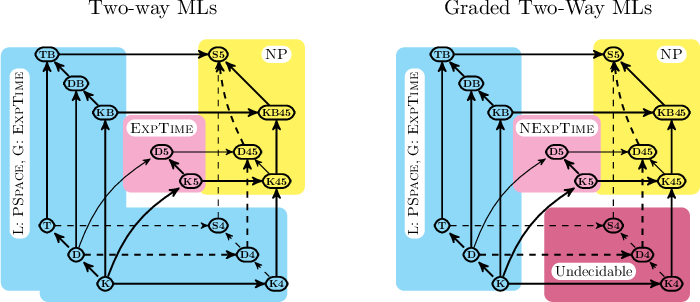}
\caption{Complexities of two-way modal logics. All bounds are tight.}
\label{fig:figure2}
\end{figure}

In the case of non-graded two-way modal language, over most relevant
classes of frames, tight complexity bounds for local and global
satisfiability are also known. The notable exceptions are global
satisfiability problems of the logics of transitive
frames, \MLKFour$(\Diamond, \RevDiamond)$, \MLSFour
$(\Diamond, \RevDiamond)$, \MLDFour $(\Diamond, \RevDiamond)$, which
are known to be in~\ExpTime{} (due to a result in~\cite{DemriN05} or due to a
translation to description logic~$SI$, whose satisfiability is
in~\ExpTime{}~\cite{Tobies01}). However, according to the survey part
of~\cite{Zolin17}, the corresponding lower bounds are missing. In the literature we were also not able find a tight lower bound for the logics of Euclidean frames,~\MLKFive $(\Diamond, \RevDiamond)$,~\MLDFive $(\Diamond, \RevDiamond)$.
We provide both missing bounds in Section~\ref{s:nongraded}, obtaining them by reductions from the acceptance problem for polynomially space bounded alternating Turing machines.\footnote{As explained to the first author by Emil Je\v{r}\'abek, the latter bound can be alternatively proved by a reduction from \MLTB{}, whose \ExpTime-hardness follows from~\cite{ChenL94}.} See the left part of Fig.~\ref{fig:figure2}
for a complete complexity map in this case.

Let us now turn our attention to the most expressive two-way graded modal language with both graded forward and graded converse modalities
 (the right part of
Fig.~\ref{fig:figure2}).
Its local and global satisfiability problems over the class of
all frames (K) are known to be, resp., \PSpace-complete and \ExpTime-complete (see the survey part of~\cite{Zolin17} and references therein).
In Section~\ref{subsec:st}, we explain how to obtain these bounds, as well as the same bounds in all cases involving neither transitivity nor Euclideanness. 
For the \ExpTime-bound, we  employ the so-called \emph{standard translation}. Over \MLKFour, \MLDFour{} and  \MLSFour{} the logics turn out to be undecidable~\cite{Zolin17}. We remark that these are the
only undecidable members of  the whole family of logics considered in this paper. What remains are the classes of frames involving the Euclidean property.
We solve them in Section~\ref{sec:eucliean_logics}. We prove that the logics~\MLKFive $(\Diamond_{\ge}, \RevDiamond_{\ge})$ and~\MLDFive $(\Diamond_{\ge}, \RevDiamond_{\ge})$ are locally and globally~\NExpTime-complete. Interestingly, 
this is a higher complexity  than the \ExpTime-complexity of the language without graded modalities~\cite{DemriN05} 
and \NP-complexity of the language without converse~\cite{KazakovP09} over the same classes of frames. We also show that, when, additionally, transitivity is required, that is, for the logics~\MLKFourFive $(\Diamond_{\ge}, \RevDiamond_{\ge})$ and~\MLDFourFive $(\Diamond_{\ge}, \RevDiamond_{\ge})$, the complexity drops down to~\NP.

Finally, we consider the above-mentioned intermediate language~$(\Diamond_{\ge}, \RevDiamond)$ in which we can count the successors, we have the basic converse modality, but we cannot count the predecessors. Our main result here, presented in Section~\ref{s:transitive}, is 
demonstrating the (local and global) \emph{finite model property} for the logics of transitive frames~$\MLKFour$,~$\MLDFour$ and~$\MLSFour$:
whenever a formula is (locally or, resp., globally) satisfiable it is (locally, resp.~globally) satisfiable over a finite frame. This implies the decidability of the (local and global) satisfiability problem (as well as the \emph{finite} satisfiability problem, in which the attention is restricted to finite frames) for these logics and thus
 solve an open problem posed in~\cite{Zolin17}. 
An analogous problem was  formulated also in the richer
 setting of description logics~\cite{KazakovSZ07},~\cite{GBasulto17}, where the corresponding logic is called $\mathcal{SIQ}^-$. That problem only recently was also positively solved~\cite{GGI19}. The results from \cite{GGI19} (which we will discuss in more details in a moment) 
allow us to derive the precise \TwoExpTime-complexity bounds for 
the logics~$\MLKFour(\Diamond_{\ge}, \RevDiamond)$,~$\MLDFour(\Diamond_{\ge}, \RevDiamond)$ and~$\MLSFour(\Diamond_{\ge}, \RevDiamond)$.
The logics of the remaining classes of 
frames retain their complexities from the graded two-way case, so the picture is as in the right part of Fig.~\ref{fig:figure2}, but the
word ''Undecidable'' should be replaced by ''\TwoExpTime''.

Due to a large number of papers in which the complexity bounds from
Fig.~\ref{fig:figure1} and Fig.~\ref{fig:figure2} are scattered, we
have not referenced all of them in this introduction. Readers wishing
to find an appropriate reference are recommended to use an online tool
prepared by the first author
(\url{bartoszjanbednarczyk.github.io/mlnavigator}).

\smallskip
\noindent
\emph{Related formalisms.} Graded modalities are examples of \emph{counting quantifiers} which are present in various formalisms.
First of all, counting quantifiers were introduced for first-order logic:~$\exists^{\ge n} x \varphi$ means: ''at least $n$ elements $x$ satisfy $\varphi$''. The satisfiability problem for some fragments of first-order logic with counting quantifiers was shown to be decidable. In particular, the two-variable fragment is \NExpTime-complete~\cite{PHartmann05}, the two-variable 
guarded fragment is \ExpTime-complete~\cite{PHartmann07}, and the one-variable fragment is \NP-complete~\cite{PHartmann08}.  We will employ the second of those results in our paper.

Counting quantifiers are also
present, in the form of the so-called \emph{number restrictions}, in some description logics, DLs. As some standard  DLs embed
modal logics (c.f. a result in~\cite[Section~2.6.2]{BHLS2017}), results on DLs with number restrictions may be used to infer upper bounds on the complexity of some graded modal logics. 

The description logic which is particularly interesting from our point of view is the already mentioned logic~$\mathcal{SIQ}^-$.
Syntactically, it can be seen as a \emph{multi-modal} logic, that is a logic whose frames interpret not just one but many accessibility relations,
with different modalities associated with these relations. 
In the case of~$\mathcal{SIQ}^-$ each of the accessibility relations can be independently required to be transitive or not. 
Recently the \emph{knowledge base satisfiability} problem for this
logic was shown decidable and \TwoExpTime-complete  \cite{GGI19}. As we said, from this result the decidability and  \TwoExpTime{} complexity  of both local and
global satisfiability of $\MLKFour(\Diamond_{\ge}, \RevDiamond)$, $\MLSFour(\Diamond_{\ge}, \RevDiamond)$ and $\MLDFour(\Diamond_{\ge}, \RevDiamond)$ can be inferred. 
Nevertheless, our proof of the finite model property for these logics remains valuable
as in \cite{GGI19} the decidability of the finite model reasoning for~$\mathcal{SIQ}^-$ is left open
(with the exception of the case in which there is only one accessibility relation and this relation
is transitive; in this case, however, our finite model construction is used and cross-referred there).

In this context it is worth noting that the logic \MLK$(\Diamond_{\ge}, \RevDiamond)$ (with the accessibility relation  not
necessarily being transitive) and the logic
\MLKFour$(\Diamond^1_{\ge}, \RevDiamond^1, \Diamond^2_{\ge}, \RevDiamond^2)$ (the bi-modal variant of $\MLKFour(\Diamond_{\ge}, \RevDiamond)$ with two independent transitive accessibility relations)
do not have the \emph{global} finite model property. 
Both these logics are contained in $\mathcal{SIQ}^-$.
An example~\MLK$(\Diamond_{\ge}, \RevDiamond)$ formula which is globally satisfiable (e.g., over an infinite
binary tree with reversed edges) but has no finite models is~$\RevDiamond p \wedge \RevDiamond \neg p \wedge \Diamond_{\le 1} \top$.
This example can be easily adapted to \MLKFour$(\Diamond^1_{\ge}, \RevDiamond^1,  \Diamond^2_{\ge}, \RevDiamond^2)$. 
On the other hand, \MLK$(\Diamond_{\ge}, \RevDiamond)$ \emph{does} have the  \emph{local} finite model property, as it is a fragment of the description logic
$\mathcal{ALCIQ}$, whose local finite model property was shown in \cite{Tob01}. The status of the local finite model property for the multi-modal
variants of \MLKFour$(\Diamond_{\ge}, \RevDiamond)$ is open.

\smallskip

\emph{Plan of the paper.} In Section \ref{s:preliminaries} we formally define the relevant modal languages and their semantics, recall the so called standard translation
and use it to derive some initial results. In Sections~\ref{sec:eucliean_logics} and~\ref{s:transitive} we investigate the classes of Euclidean frames and, respectively, transitive frames.
Finally, in Section~\ref{s:nongraded} we provide two lower bounds filling the gaps in the classification of the complexity of non-graded languages.

This work is an extended version of our conference paper~\cite{BednarczykK019}. 



\section{Preliminaries} \label{s:preliminaries}

\subsection{Languages, Kripke structures and satisfiability} \label{subsec:preliminaries-lkss} 

Let us fix a countably infinite set~$\Pi$ of \emph{propositional variables}. 
The \emph{language} of graded two-way modal logic is defined inductively 
as the smallest set of formulas containing~$\Pi$, closed under
Boolean connectives and, for any formula~$\varphi$, 
containing~$\Diamond_{\ge n} \varphi$ and~$\RevDiamond_{\ge n} \varphi$, 
for all~$n \in \N$. For a given formula~$\varphi$, we denote its \emph{length} 
with~$|\varphi|$, and measure it as the number of symbols required to 
write~$\varphi$, with numbers in subscripts~$_\ge n$ encoded in binary 
(i.e., encoding a number~$n$ requires~$\log{n}$ bits rather than~$n$ bits).

The basic modality~$\Diamond$ can be defined in terms of graded 
modalities:~$\Diamond \varphi := \Diamond_{\ge 1} \varphi$.
Analogously, for the converse modality:~$\RevDiamond := \RevDiamond_{\ge 1}$.
Keeping this in mind, we may treat all languages mentioned in the introduction 
as fragments of the above-defined graded two-way modal language. 
We remark that we may also introduce other modalities, e.g.
$$
\Diamond_{\le n} \varphi := \neg \Diamond_{\ge n+1} \varphi, \;
\RevDiamond_{\le n} \varphi := \neg \RevDiamond_{\ge n+1} \varphi, \; 
\Box \varphi := \neg \Diamond \neg \varphi, \; \text{and} \;
\RevBox \varphi := \neg \RevDiamond \neg \varphi.
$$
The semantics is defined with respect to \emph{Kripke structures}, that is, 
structures over the relational signature composed of unary predicates~$\Pi$ and 
with a binary predicate~$R$. Such structures are represented as 
triples~$\str{A}= \langle W, R, V \rangle$, where~$W$ is the 
\emph{universe},~$R$ is a binary \emph{accessibility} relation on~$W$, 
and~$V$ is a function~$V:\Pi \rightarrow \cP(W)$ called \emph{valuation}.
Elements from the set~$W$ are often called \emph{worlds}.

The \emph{satisfaction relation}~$\models$ is defined inductively as follows:
\begin{itemize}
\item $\str{A}, w \models p$ iff $w \in V(p)$, for all~$p \in \Pi$,
\item $\str{A}, w \models \neg \varphi$ iff~$\str{A}, w \not\models \varphi$ 
and similarly for the other Boolean connectives,
\item $\str{A}, w \models \Diamond_{\ge n} \varphi$ iff there are $\geq n$ 
worlds~$v \in W$ such that~$\langle w, v \rangle \in R$
and~$\str{A}, v \models \varphi$
\item $\str{A}, w \models \RevDiamond_{\ge n} \varphi$ iff there are $\geq n$ worlds~$v \in W$ such that~$\langle v, w \rangle \in R$
and~$\str{A}, v \models \varphi$.
\end{itemize}

For a given Kripke structure~$\str{A}=\langle W, R, V \rangle$ 
we call the pair~$\langle W, R \rangle$ its \emph{frame}. For a class of 
frames~$\cF$, we define the local (global) satisfiability problem of a modal 
language~$\cL$ over~$\cF$ (or equivalently for a modal logic~$\cF(\cL^*)$) as 
follows: given a formula~$\varphi$ from a language~$\cL$, verify whether~$\varphi$ 
is satisfied at some world (all worlds)~$w$ of some structure~$\str{A}$ whose 
frame belongs to~$\cF$. 

We announced in the introduction that we are interested in 
classes of frames characterized by any positive combination 
of the axioms of reflexivity~(T), seriality~(D), symmetry~(B), 
transitivity~(4), and the Euclidean property~(5), recalled below.
$$
	\begin{array}{cl@{\quad}l}
	(\ML{D}) & \mbox{seriality}   & \forall x\exists y\,(xRy)\\
	(\ML{T}) & \mbox{reflexivity} & \forall x\,(xRx) \\
	(\ML{B}) & \mbox{symmetry}  & \forall xy\,(xRy\Rightarrow yRx) \\
	(\ML{4}) & \mbox{transitivity} & \forall xyz\,(xRy\land yRz\Rightarrow xRz) \\
	(\ML{5}) & \mbox{Euclideanness} & \forall xyz\,(xRy\land xRz \Rightarrow yRz)\\
	\end{array}
$$ 
We say that a modal logic~$\cF(\cL^*)$ has the \emph{finite local (global) 
model property} if any formula of~$\cL$ which is satisfied in some world 
(all worlds) of some structure from~$\cF$ is also satisfied in some world 
(all worlds) of a \emph{finite} structure from~$\cF$.

\subsection{Standard translations} \label{subsec:st}

Modal logic can be seen as a fragment of first-order logic via the so-called 
\emph{standard translation} (see e.g.,~\cite{BBV01}). Here we present its 
variation tailored for graded and converse modalities and discuss how it can be
used to establish exact complexity bounds for some of graded two-way modal logics.

In the forthcoming definition, we define a function~$\stv$ for~$v \in \{x,y\}$, 
which takes an input two-way modal logic formula~$\varphi$ and returns 
an equisatisfiable first-order formula. Definitions of~$\stx$ and~$\sty$
are symmetric, hence we present the definition of~$\stx$ only.

\begin{gather}
\stx(p) = p(x) \text{ for all~$p \in \Pi$} \\
\stx(\varphi \wedge \psi ) = \stx(\varphi) \wedge \stx(\psi)  \text{ similarly for~$\neg$,~$\vee$, etc.}\\
\stx(\Diamond_{\geq n}\varphi) = \exists_{\geq n}.y(R(x,y) \wedge \sty(\varphi))\\
\stx(\RevDiamond_{\geq n}\varphi) = \exists_{\geq n}.y(R(y, x) \wedge \sty(\varphi))
\end{gather}  
Translated formulas lie in the two-variable guarded fragment of first-order
logic extended with	 counting quantifiers~$\GCTwoo$. Observe that a modal 
formula~$\varphi \in \cL$ is (finitely) locally-satisfiable iff a 
formula~$\exists{x} \; \stx(\varphi) \in \GCTwoo$ is (finitely) satisfiable 
and that~$\varphi$ is (finitely) globally-satisfiable
iff~$\forall{x} \; \stx(\varphi) \in \GCTwoo$ is (finitely) satisfiable.
Since definitions of symmetry, seriality and reflexivity, as recalled in 
the previous section, are~$\GCTwoo$ formulas, the standard translation can be 
used to provide a generic upper bound for the 
logics~$\cF(\Diamond_{\ge}, \RevDiamond_{\ge})$ over all 
classes of frames~$\cF$ involving neither transitivity nor Euclideanness. 
From the fact that the global satisfiability problem is~$\ExpTime$--hard 
even for the basic modal language~$\cF(\Diamond)$~\cite{Blackburn06} 
and from~$\ExpTime$-completeness of~$\GCTwoo$~\cite{PHartmann07}, 
we conclude the following theorem:
\begin{theorem}
  The global satisfiability problem 
  for~$\cF(\Diamond_{\ge}, \RevDiamond_{\ge})$ where~$\cF$ is any class of 
  frames from the modal cube involving neither 
  transitivity nor Euclideanness, is~$\ExpTime$-complete.
\end{theorem}

For the local satisfiability problem, its complexity decreases to~$\PSpace$. 
For two-way graded language over~$\MLK$, $\MLD$ and~$\MLT$, we can simply adapt 
an existing tableaux algorithm by Tobies~\cite{Tobies01}, which yields a 
tight~$\PSpace$ bound. Moreover, if a class of frames is symmetric, then 
forward and converse modalities coincide and thus we may simply apply the 
result on graded one-way languages from~\cite{KazakovP09}. 
The~$\PSpace$ lower bounds for the above-mentioned logics are inherited from the basic 
modal logic~$\MLK$~\cite{Lad77} and hold even in the case of their propositional-variable-free fragment\cite{ChagrovR02}.
Thus we can conclude the following.

\begin{theorem}
The local satisfiability problem for~$\cF(\Diamond_{\ge}, \RevDiamond_{\ge})$, 
where~$\cF$ is any class of frames from the modal cube involving neither 
transitivity nor Euclideanness is~$\PSpace$-complete.
\end{theorem}



\section{Euclidean frames: counting successors and predecessors} \label{sec:eucliean_logics}

This section is dedicated to modal languages over the classes of frames
satisfying Euclideanness. We demonstrate an exponential 
gap~($\NExpTime$ versus~$\NP$) in the complexities of modal logics over 
Euclidean frames ($\MLKFive$ and~$\MLDFive$) 
and modal logics over transitive Euclidean frames ($\MLKFourFive$ and~$\MLDFourFive$).

The two remaining Euclidean logics of our interest, namely~$\MLKBFourFive$ and~$\MLSFive$, 
whose frames are additionally symmetric, may be seen as one-way logics (as~$\RevDiamond_\ge$ can
be always replaced by~$\Diamond_\ge$).
Hence, their~$\NP$ upper bounds follows from previous works on one-way MLs~\cite{KazakovP09}.
The lower bound is inherited from the Boolean satisfiability problem~\cite{Cook71}. 
Thus:
\begin{theorem}[Consequence of~\cite{KazakovP09}.]
The local satisfiability and the global satisfiability problems 
for modal logics~$\MLKBFourFive(\Diamond_{\ge}, \RevDiamond_{\ge})$ 
and~$\MLSFive(\Diamond_{\ge}, \RevDiamond_{\ge})$
are~$\NP$-complete.
\end{theorem}

\subsection{The shape of Euclidean frames}

We start by describing the shape of frames under consideration. 
Let~$\str{A}$ be a \emph{Euclidean structure}, i.e., a 
Kripke structure~$\str{A} = \tuple{W, R, V}$ whose accessibility 
relation~$R$ satisfies the Euclidean property.

A world~$w \in W$ is called a \emph{lantern}, if~$\tuple{w', w} \not\in R$ 
holds for every~$w' \in W$. The set of all lanterns in~$\str{A}$ is denoted 
with~$L_{\str{A}}$. We say that lantern~$l \in W$ \emph{illuminates} 
a world~$w\in W$, if~$\tuple{l, w} \in R$ holds. The previous definition is 
lifted to the sets of worlds in an obvious way: a lantern~$l$ \emph{illuminates 
a set of worlds}~$I \subseteq W$ if~$l$ illuminates every world~$w$ from~$I$. 

We say that two worlds~$w_1, w_2 \in W$ are \emph{$R$-equivalent} 
(or simply \emph{equivalent} if~$R$ is known from the context), if 
both~$\tuple{w_1,w_2} \in R$ and~$\tuple{w_2, w_1}\in R$ holds. 
The~\emph{$R$-clique} for a world~$w_1$ in a structure~$\str{A}$ is the 
set~$Q_{\str A}(w_1) \subseteq W$ consisting of~$w_1$ together with
all of its~$R$-equivalent worlds. With~$Q_{\str{A}}$ we denote the 
set~$W \setminus L_{\str A}$ of \emph{inner} (i.e. non-lantern) \emph{worlds}. 
See Fig.~\ref{fig:figure3} for a drawing of an example Euclidean structure.

\begin{figure}[h]
\begin{center}
 
\tikzset{
pattern size/.store in=\mcSize, 
pattern size = 5pt,
pattern thickness/.store in=\mcThickness, 
pattern thickness = 0.3pt,
pattern radius/.store in=\mcRadius, 
pattern radius = 1pt}
\makeatletter
\pgfutil@ifundefined{pgf@pattern@name@_qjwfuwfqk}{
\makeatletter
\pgfdeclarepatternformonly[\mcRadius,\mcThickness,\mcSize]{_qjwfuwfqk}
{\pgfpoint{-0.5*\mcSize}{-0.5*\mcSize}}
{\pgfpoint{0.5*\mcSize}{0.5*\mcSize}}
{\pgfpoint{\mcSize}{\mcSize}}
{
\pgfsetcolor{\tikz@pattern@color}
\pgfsetlinewidth{\mcThickness}
\pgfpathcircle\pgfpointorigin{\mcRadius}
\pgfusepath{stroke}
}}
\makeatother
\tikzset{every picture/.style={line width=0.75pt}} 

\begin{tikzpicture}[x=0.75pt,y=0.75pt,yscale=-1,xscale=1]

\draw  [pattern=_qjwfuwfqk,pattern size=10pt,pattern thickness=0.75pt,pattern radius=0.75pt, pattern color={rgb, 255:red, 0; green, 0; blue, 0}] (80,129.25) .. controls (80,105.64) and (115.03,86.5) .. (158.25,86.5) .. controls (201.47,86.5) and (236.5,105.64) .. (236.5,129.25) .. controls (236.5,152.86) and (201.47,172) .. (158.25,172) .. controls (115.03,172) and (80,152.86) .. (80,129.25) -- cycle ;
\draw   (145,129) .. controls (145,117.95) and (160.67,109) .. (180,109) .. controls (199.33,109) and (215,117.95) .. (215,129) .. controls (215,140.05) and (199.33,149) .. (180,149) .. controls (160.67,149) and (145,140.05) .. (145,129) -- cycle ;
\draw   (100,130) .. controls (100,118.95) and (115.67,110) .. (135,110) .. controls (154.33,110) and (170,118.95) .. (170,130) .. controls (170,141.05) and (154.33,150) .. (135,150) .. controls (115.67,150) and (100,141.05) .. (100,130) -- cycle ;
\draw    (136.5,36) -- (100.72,128.14) ;
\draw [shift={(100,130)}, rotate = 291.22] [color={rgb, 255:red, 0; green, 0; blue, 0 }  ][line width=0.75]    (10.93,-3.29) .. controls (6.95,-1.4) and (3.31,-0.3) .. (0,0) .. controls (3.31,0.3) and (6.95,1.4) .. (10.93,3.29)   ;

\draw    (187.5,40) -- (145.86,127.2) ;
\draw [shift={(145,129)}, rotate = 295.53] [color={rgb, 255:red, 0; green, 0; blue, 0 }  ][line width=0.75]    (10.93,-3.29) .. controls (6.95,-1.4) and (3.31,-0.3) .. (0,0) .. controls (3.31,0.3) and (6.95,1.4) .. (10.93,3.29)   ;

\draw    (136.5,36) -- (169.33,128.12) ;
\draw [shift={(170,130)}, rotate = 250.38] [color={rgb, 255:red, 0; green, 0; blue, 0 }  ][line width=0.75]    (10.93,-3.29) .. controls (6.95,-1.4) and (3.31,-0.3) .. (0,0) .. controls (3.31,0.3) and (6.95,1.4) .. (10.93,3.29)   ;

\draw    (187.5,40) -- (214.41,127.09) ;
\draw [shift={(215,129)}, rotate = 252.82999999999998] [color={rgb, 255:red, 0; green, 0; blue, 0 }  ][line width=0.75]    (10.93,-3.29) .. controls (6.95,-1.4) and (3.31,-0.3) .. (0,0) .. controls (3.31,0.3) and (6.95,1.4) .. (10.93,3.29)   ;

\draw    (136.5,36) -- (135.03,128) ;
\draw [shift={(135,130)}, rotate = 270.90999999999997] [color={rgb, 255:red, 0; green, 0; blue, 0 }  ][line width=0.75]    (10.93,-3.29) .. controls (6.95,-1.4) and (3.31,-0.3) .. (0,0) .. controls (3.31,0.3) and (6.95,1.4) .. (10.93,3.29)   ;
\draw [shift={(136.5,36)}, rotate = 90.91] [color={rgb, 255:red, 0; green, 0; blue, 0 }  ][fill={rgb, 255:red, 0; green, 0; blue, 0 }  ][line width=0.75]      (0, 0) circle [x radius= 3.35, y radius= 3.35]   ;
\draw    (187.5,40) -- (180.17,127.01) ;
\draw [shift={(180,129)}, rotate = 274.82] [color={rgb, 255:red, 0; green, 0; blue, 0 }  ][line width=0.75]    (10.93,-3.29) .. controls (6.95,-1.4) and (3.31,-0.3) .. (0,0) .. controls (3.31,0.3) and (6.95,1.4) .. (10.93,3.29)   ;
\draw [shift={(187.5,40)}, rotate = 94.82] [color={rgb, 255:red, 0; green, 0; blue, 0 }  ][fill={rgb, 255:red, 0; green, 0; blue, 0 }  ][line width=0.75]      (0, 0) circle [x radius= 3.35, y radius= 3.35]   ;

\draw (138,22) node   {$l_{1}$};
\draw (190,24) node   {$l_{2}$};
\draw (255,129) node   {$Q_{\mathfrak{A}}$};

\end{tikzpicture}
\end{center}
\caption{A Euclidean structure~$\str{A}$ with lanterns~$L_{\str{A}} = \{l_1, l_2\}$}
\label{fig:figure3}
\end{figure}
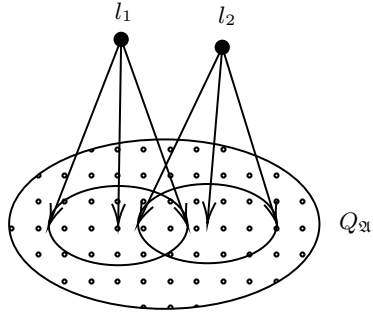
It is easy to observe that for any world~$w_1 \in W$, all members of the 
clique~$Q_{\str A}(w_1)$ are~$R$-equivalent. This justifies why we have chosen 
the term ``clique'' to name such sets. 
\begin{observation}
Any distinct worlds~$w', w''$ from the~$R$-clique~$Q_{\str A}(w)$ of~$w$ are~$R$-equivalent.
\end{observation}
\begin{proof} 
From the definition of~$R$-equivalence, we know that both~$\tuple{w, w'} \in R$ 
and~$\tuple{w, w''} \in R$ hold. Since the relation~$R$ satisfies the Euclidean 
property we infer that~$\tuple{w', w''} \in R$ holds 
and~$\tuple{w'', w'} \in R$ holds, which implies~$R$-equivalence of~$w$ and~$w'$.
\end{proof}

An immediate conclusion from the above observation is that the 
equality~$Q_{\str A}(w) = Q_{\str A}(w_1)$ holds for any world~$w \in Q_{\str A}(w_1)$.
Thus we will say that~$Q$ is an~\emph{$R$-clique in $\str{A}$} 
if the equality~$Q = Q_{\str A}(w_1)$ holds for some (equivalently: for any) 
world~$w_1 \in Q$.

As usual in modal logics, we can restrict our attention
to~\emph{$R$-connected} models, that is those
models~$\str{A} = \tuple{W,R,V}$ for which~$\tuple{W, R \cup R^{-1}}$
is a connected graph. The following lemma describes the shape of
Euclidean structures under consideration. It is very similar to
Lemma~2 in~\cite{KazakovP09}.

\begin{lemma} \label{lem:PropOfEuclStr} 
If $\str{A}$ is an $R$-connected structure over a Euclidean 
frame~$\tuple{W, R}$, then all worlds~$w$ in~$Q_{\str A}$ are \emph{reflexive} 
(i.e~$\tuple{w, w} \in R$ holds) and~$Q_{\str A}$ is an~$R$-clique.
\end{lemma}

\begin{proof}
  In the course of proof we will refer to the formula (\textbf{5})
  from Section~\ref{subsec:preliminaries-lkss}, which defines
  Euclidean property. We will show that all worlds in~$Q_{\str A}$ are
  reflexive and all worlds in~$Q_{\str A}$ are~$R$-equivalent. To show
  reflexivity take any~$w \in Q_{\str A}$. By definition
  of~$Q_{\str A}$ there exists~$w' \in W$ such
  that~$\tuple{w', w} \in R$. Since~$\str A$ satisfies (\textbf{5}),
  by taking~$w'$ as~$x$ and~$w$ as both~$y$ and~$z$ in (\textbf{5}),
  we infer~$\tuple{w, w} \in R$.
  
  To demonstrate $R$-equivalence we will employ some simple
  observations. First, the relation
  $R \cap (Q_{\str A} \times Q_{\str A})$ is symmetric.  To prove it
  take any $w_1, w_2 \in Q_{\str A}$ with $\tuple{w_1,w_2} \in
  R$. Then, use reflexivity of $w_1$ and the Euclidean property
  (with~$w_1$ taken as~$x$ and~$z$ simultaneously, and~$w_2$ as~$y$ in
  (\textbf{5})) to infer~$\tuple{w_2, w_1} \in R$.  Second,
  $R \cap (Q_{\str A}\times Q_{\str A})$ is transitive. To prove it
  take any $w_1, w_2, w_3 \in Q_{\str A}$ with $\tuple{w_1,w_2} \in R$
  and $\tuple{w_2,w_3} \in R$. Symmetry of
  $R \cap (Q_{\str A}\times Q_{\str A})$ gives us
  $\tuple{w_2,w_1} \in R$. Then, by the Euclidean property (with~$w_2$
  taken as~$x$, $w_3$ taken as $z$, and~$w_1$ as~$y$ in (\textbf{5}))
  we infer $\tuple{w_1,w_3} \in R$. Third, if $\tuple{l, w_1} \in R$
  and $\tuple{l, w_2} \in R$, for some $l \in L$ and
  $w_1, w_2 \in Q_{\str A}$, then $\tuple{w_1,w_2} \in R\cup
  R^{-1}$. This observation again simply follows from (\textbf{5}).

  Now take any $w \in Q_{\str A}$. We will show that
  $Q_{\str A} = Q_{\str A}(w)$, i.e. that $Q_{\str A}$ is the R-clique
  for $w$. Take any $w' \in Q_{\str A}$. We will show that both
  $\tuple{w,w'} \in R$ and $\tuple{w',w} \in R$. Since $\str A$ is
  connected, there exists a $(R \cup R^{-1}$)-path from $w$ to $w'$ in
  $\str A$. By inductive application of the third observation above we
  may assume that all elements of the path belong to $Q_{\str
    A}$. Then by the first observation (symmetry) we may assume that this
  is actually an $R$-path. Then, by the second observation
  (transitivity) the path reduces to a single edge
  $\tuple{w,w'} \in R$. In the same way we may show that
  $\tuple{w',w} \in R$. Thus all worlds in $Q_{\str A}$ are
  $R$-equivalent with $w$. Since all other worlds in $\str A$ are
  lanterns, they cannot be $R$-equivalent with $w$. Thus $Q_{\str A}$
  is indeed the R-clique for $w$.
\end{proof}

\subsection{The universal modality}

Before we start proving complexity results for the family of Euclidean logics,
we show that global and local satisfiability problems are inter-reducible over
any class of frames involving the Euclidean property. 

Having restricted our attention to~$R$-connected models, we will show 
that the \emph{universal modality}~$\U$ can be defined in terms of 
standard~(i.e,~$\Diamond$ and~$\RevDiamond$) modalities. 
Recall that the semantics of~$\U\varphi$ is defined as 
follows:~$\str{A}, w \models \U\varphi$, iff for every  world~$x$ the 
condition~$\str{A}, x \models \varphi$ holds. Taking a look at the shape of 
Euclidean structures (see e.g. Lemma~\ref{lem:PropOfEuclStr}), it is not
difficult to see that to propagate satisfaction of a given formula~$\varphi$ 
through the whole structure, it is sufficient to first traverse all inner 
elements and from each of them propagate the satisfaction of~$\varphi$ to 
their predecessors. This intuition can be formalised by 
taking~$\U\varphi := \varphi \wedge \Box\Box\RevBox\varphi$. 

\begin{lemma} \label{lem:UniversalModalityDefinableEucl}
  Let~$\str{A} = \tuple{W, R, V}$ be an $R$-connected Euclidean
  structure.
  Then~$\str{A}, w_0 \models \varphi \wedge \Box\Box\RevBox\varphi$
  holds for some world~$w_0 \in W$ iff $\str{A}, v \models \varphi$ holds for all worlds~$v \in W$.
\end{lemma}
\begin{proof}
\newcommand{\wi}{w_{\textit{in}}}
\newcommand{\wa}{w_{\textit{ax}}}
Let~$\str{A} = \tuple{W, R, V}$ be an $R$-connected Euclidean structure
and let~$\str{A},w \models \varphi \wedge \Box\Box\RevBox\varphi$ hold for
some world~$w_0 \in W$. We will show that it implies that~$\varphi$ 
is true in every world~$w \in W$ (the opposite direction of the Lemma is trivial).

First, if $R = \emptyset$ then $\str{A}$ is a singleton structure,
because it is $R$-connected. In this case the implication trivially
holds. So, assume that $R \neq \emptyset$. Define
$S = R \circ R \circ R^{-1}$. We will show that $S$ is the universal
relation $W\times W$.  Indeed, take any $a, b \in W$. Then there
exists $x \in Q_{\str A}$ such that $R(a,x)$ holds (if
$a \in Q_{\str A}$ then, by Lemma~\ref{lem:PropOfEuclStr}, $a$ is
reflexive, so take $x = a$; if $a \in L_{\str A}$, such an $x$
exists, since $\str A$ is connected). Similarly, there exists
$y \in Q_{\str A}$ such that $R(b,y)$. Now we have $R(x,y)$, since $R$
is universal on $Q_{\str A}$ by Lemma~\ref{lem:PropOfEuclStr}. Thus we
have $R(a, x)$, $R(x,y)$ and $R^{-1}(y, b)$, so $S(a,b)$ holds and
thus $S = W \times W$. Therefore
$\str{A}, w_0 \models \Box\Box\RevBox\varphi$ implies
$\str{A}, v \models \Box\Box\RevBox\varphi$, for any $v \in W$.
\end{proof}

We now argue that the local and global satisfiability problems coincide for modal logics over Euclidean frames.

\begin{lemma} \label{lem:locisglobforeuclogics}
Let~$(\cL, \cF)$ be a modal logic whose language contains~$\Diamond$ and~$\RevDiamond$ and
$\cF$ is a class of frames from the modal cube satisfying the  Euclidean property. Then the global satisfiability problem 
for~$\cF(\cL^*)$ is~$\LogSpace$ reducible to the local satisfiability problem 
for~$\cF(\cL^*)$ and vice-versa.
\end{lemma}
\begin{proof}
  As usual for modal logics we may restrict to satisfiability over
  connected structures.  Since~$\cF$ is Euclidean and we have
  both~$\Diamond$,~$\RevDiamond$ at our disposal, we know that the
  universal modality~$\U$ is definable in~$\cF(\cL^*)$ (see:
  Lemma~\ref{lem:UniversalModalityDefinableEucl}).  From the semantics
  of~$\U$ we can immediately conclude that any modal
  formulas~$\varphi_l, \varphi_g$ the following equivalences
  hold:~$\varphi_l$ is locally-satisfiable iff~$\neg\U\neg\varphi_l$
  is globally-satisfiable and~$\varphi_g$ is globally-satisfiable
  iff~$\U\varphi_g$ is locally-satisfiable.
\end{proof}

\subsection{The upper bound for graded two-way~\texorpdfstring{$\MLKFive$}{K5} and~\texorpdfstring{$\MLDFive$}{D5}}

This Section is dedicated to the following theorem.

\begin{theorem} \label{thm:ModEuclNExp}
  The local and global satisfiability problems for Euclidean Modal Logics~$\MLKFive(\Diamond_{\geq}, \RevDiamond_{\geq})$
	and~$\MLDFive(\Diamond_{\geq}, \RevDiamond_{\geq})$
  are in \NExpTime.
\end{theorem}
\begin{proof}
  Note that here we may again restrict to satisfiability over
  connected frames.  We start with the case of the class of all
  Euclidean frames~$\MLKFive$.  We translate a given modal
  formula~$\varphi$ to the two-variable logic with counting \CTwo, in
  which both graded modalities and the shape of connected Euclidean
  structures, as defined in Lemma~\ref{lem:PropOfEuclStr}, can be
  expressed. Since satisfiability of~\CTwo is in
  \NExpTime~\cite{PHartmann07}, we obtain the desired conclusion.
  Recall the standard translation~$\st$ from
  Section~\ref{subsec:st}. Let~$\lan(\cdot)$ be a new unary predicate
  and define~$\varphi_{\tr}$ as
  $$
    \st_x(\varphi) \wedge \forall{x}\forall{y}.\left( \neg\lan(x) \wedge
    \neg\lan(y) \rightarrow R(x,y)\right) \wedge \left(\lan(y) \rightarrow
    \neg R(x, y)\right).
  $$
  Since~$\st_x(\varphi)$ belongs to \GCTwo{},~$\varphi_{\tr}$
  belongs to \CTwo (but not to \GCTwo), and has  one free variable~$x$. Let~$\str B$ be a Kripke structure over a Euclidean frame. Expand~$\str B$ to a structure~$\str B^{+}$ by setting~$\lan^{\str B^{+}} =\{ w \in \str B \mid w \in L_{\str B}\}$.
  Taking into account Lemma~\ref{lem:PropOfEuclStr} a structural
  induction on~$\varphi$ easily establishes the following condition
$$ \str B, w_0 \models \varphi\text{ if and only if } \str B^{+} \models 
  \varphi_{\tr}[w_0 /x] \text{ for every world~$w_0 \in B$}.$$
  Thus, a~$\MLKFive(\Diamond_{\geq}, \RevDiamond_{\geq})$ formula~$\varphi$ is locally satisfiable if and only if the \CTwo formula~$\exists_{\geq 1}{x}. \varphi_{\tr}$ is satisfiable, yielding a
  \NExpTime{} algorithm for~$\MLKFive(\Diamond_{\geq}, \RevDiamond_{\geq})$ 
  local satisfiability. Membership of  global satisfiability in \NExpTime{}
	is implied by Lemma~\ref{lem:UniversalModalityDefinableEucl}.

 For the case of serial Euclidean frames,~$\MLDFive$, 
 it suffices to supplement the \CTwo formula defined in 
the case of~$\MLKFive$ with the conjunct~$\exists{x}.(\neg\lan(x) )$ expressing
  seriality. Correctness follows then from the 
simple observation that  a Euclidean frame is serial iff it contains
  at least one non-lantern world (recall that all these worlds are
  reflexive).
\end{proof} 

\subsection{Lower bounds for two-way graded~\texorpdfstring{$\MLKFive$}{K5} and~\texorpdfstring{$\MLDFive$}{D5}}

We now show a matching~$\NExpTime$-lower bound for the logics from the
previous section. We concentrate on local satisfiability, but by Lemma~\ref{lem:UniversalModalityDefinableEucl}
the results will hold also for global satisfiability.
Actually, we obtain a
stronger result, namely, we show that the two-way graded modal logics~$\MLKFive$ 
and~$\MLDFive$ remain~$\NExpTime$-hard even if counting in one-way
(either backward or forward) is forbidden. Hence, we show hardness of the logics~$\MLKFive (\Diamond_{\ge}, \RevDiamond)$ and~$\MLDFive (\Diamond_{\ge}, \RevDiamond)$.
We recall that 
this gives a higher complexity  than the~\ExpTime-complexity of the language~$\Diamond, \RevDiamond$~\cite{DemriN05}
and~\NP-complexity of the language~$\Diamond_{\ge}$~\cite{KazakovP09} over the same classes of frames.

In order to prove~$\NExpTime$-hardness of the Euclidean Two-Way Graded Modal Logics~$\MLKFive$ and~$\MLDFive$ we employ
a variant of the classical tiling problem, namely~\emph{exponential torus tiling problem} from~\cite{Luc02}.

\begin{definition}[4.15 from~\cite{Luc02}] \label{def:tilings}
A \emph{torus tiling problem}~$\cP$ is a tuple~$(\cT, \cH, \cV)$, 
where~$\cT$ is a finite set of tile types and~$\cH, \cV \subseteq \cT \times \cT$ represent the horizontal and
vertical matching conditions. Let~$\cP$ be a tilling problem and~$c = t_0,t_1, \ldots, t_{n-1} \in \cT^{n}$ an initial condition.
A mapping~$\tau : \{ 0,1, \ldots, 2^n-1 \} \times 
\{ 0,1, \ldots, 2^n-1 \} \rightarrow \cT$ is a \emph{solution} for~$\cP$ and~$c$ if and only if, for all~$i,j < 2^n$, the
following holds~$(\tau(i,j), \tau(i \oplus_{2^n} 1, j)) \in \cH, (\tau(i,j), 
\tau(i, j \oplus_{2^n} 1)) \in \cV$ and~$\tau(0,i) = t_i$ for all~$i < n$, where
$\oplus_i$ denotes addition modulo~$i$. It is well-known that there exists a~$\NExpTime$-complete torus tiling problem. 
\end{definition}

\subsubsection{Outline of the proof.} The proof is based on  a polynomial time 
reduction from a torus tiling problem as in Definition~\ref{def:tilings}. Henceforward we assume that a \NExpTime-complete
torus tiling  problem~$\cP = ( \cT, \cH, \cV)$ is fixed. Let  
$c = t_0, t_1, \ldots, t_{n-1} \in \cT^n$ be its initial condition. We write a formula which is (locally) satisfiable 
iff~$\tuple{\cP, c}$ has a solution. Each cell of the torus carries a \emph{position}~$\tuple{H,V} \in \{ 0, 1, \ldots, 2^{n}-1 \} 
\times \{ 0, 1, \ldots, 2^{n}-1 \}$, encoded in binary in a natural way by 
means of propositional letters~$v_0, v_1, \ldots, v_{n-1}$ and 
$h_0, h_1, \ldots, h_{n-1}$, with~$h_0$ and~$v_0$ denoting the least significant 
bits. In the reduction, a single cell of the torus corresponds to a unique 
\emph{inner}, i.e., non-lantern, world. Since there are exactly~$2^n \cdot 2^n$ 
cells, we enforce that also the total number of inner worlds is equal 
to~$2^n \cdot 2^n$. We make use of graded modalities to specify that every inner world has 
exactly~$2^n \cdot 2^n$ successors.
We stress here that this is the only place where we employ counting. 
Thus the proof works in the case where graded converse modalities are disallowed (but the
basic converse modality will be necessary).  Alternatively we could equivalently write that every inner world has exactly~$2^n \cdot 2^n$ inner predecessors,
and obtain hardness of the language with graded converse modalities, but without graded forward modalities. 

Once we enforced a proper size of our torus, we must be sure that two
distinct inner worlds carry different positions. We do this in two
steps. We first write that a world with position~$(0,0)$ occurs in a
model. For the second step, we assume that the grid is
chessboard-like, i.e., all elements are coloured black or white in the
same way as a chessboard is. Then, we say that every world is
illuminated by four lanterns, where each of them
propagates~${\oplus_{2^n}}1$ relation on the proper axis (from a black
node to a white one and vice versa).
Finally, having the torus prepared we encode a solution for the given
tiling by simply labelling each inner world with some tile letter~$t$
and ensure (from the vantage point of the lanterns) that any two horizontal or
vertical neighbours do not violate the tiling constraints.

\subsubsection{Encoding the exponential torus.}\label{subsub:encodingTorus} Our goal is now to define a formula describing the exponential torus. The shape of the formula
is following:
$$
\phiTorus \deff
\phiFirstCell \wedge 
\U \left( 
  \phiPartition \wedge 
  \phiChessboard \wedge 
  \phiTorusSize \wedge 
  \phiSuccessors 
\right)
$$
where~$\U$ is the universal modality as in Lemma~\ref{lem:UniversalModalityDefinableEucl}.
The formula is going to say that: 
(i) the current world has position~$(0,0)$; 
(ii) every world is either a lantern or an inner world; 
(iii) the torus is chessboard-like, i.e., its cells are with coloured with~$\bl$~(black) 
and with~$\wht$~(white) exactly as a real chessboard is; 
(iv) the overall size of the torus is equal to~$2^n \cdot 2^n$; 
(v) each world of the torus has a proper vertical and a proper horizontal 
successor. The first four properties are straightforward to define:
\begin{gather*}
\phiFirstCell \deff \inn \wedge \wht \wedge 
\bigwedge_{i=0}^{n-1} \left( \neg v_i \wedge \neg h_i \right)\\
\phiPartition \deff (\lan \iff \neg \inn) \wedge 
(\lan \iff \neg \RevDiamond \top )\\
\phiChessboard \deff (\wht \iff \neg \bl) \wedge (\wht \iff (v_0 \iff h_0))\\
\phiTorusSize \deff \inn \rightarrow \Diamond_{= 2^n \cdot 2^n} \top
\end{gather*}

Note that the formula~$\phiTorusSize$ indeed expresses (iv), as the
set of all inner worlds forms a clique. The obtained formulas are of
polynomial length since the number~$2^n \cdot 2^n$ is encoded in
binary.

What remains is to define~$\phiSuccessors$. For this, for every inner
world we ensure that there exists a proper lantern responsible for establishing the appropriate successor relation. There will be four
different types of such lanterns, denoted by symbols:~$\vbw$,~$\hbw$,~$\vwb$,~$\hwb$. The intuition is the following: the
first letter \textit{h} or \textit{v} indicates whether a lantern is
responsible for an~$H$-- or~$V$--relation.  The last two letters say 
whether a successor relation
  will be established between black and white worlds, or in the
  opposite way.  
\begin{multline*}
 \phiSuccessors \deff
(\lan \rightarrow\bigvee_{\heartsuit \in \{ \vbw, \hbw, \vwb, \hwb \}} 
(\heartsuit \wedge \varphi_{\heartsuit})
) 
\wedge\\
(
\inn \rightarrow 
\bigwedge_{\heartsuit \in \{ \vbw, \hbw, \vwb, \hwb \}}
\RevDiamond ( \lan \wedge \varphi_{\heartsuit} )
)
\end{multline*}

It suffices to define formulas~$\phivbw$,~$\phihbw$,~$\phivwb$ and
$\phihwb$. Let us first define~$\phivbw$.  The formula below, intended
to be interpreted at a lantern, consists of three parts: (i) the black
and the white worlds illuminated by the lantern are pseudo-unique, i.e.,
all white (respectively, black) worlds illuminated by the same lantern
carry the same position; uniqueness will follow later from
$\phiTorusSize$; (ii) all black worlds illuminated by the lantern have
the same~$H$-position as all white worlds illuminated by this lantern;
(iii) if~$V_w$ (respectively,~$V_b$) encodes a
$V$-position of the white (respectively, black) worlds
illuminated by the lantern, then~$V_w = V_b \oplus_{2^n} 1$. 
Let us define~$\phivbw$ as:
$$\phivbw \deff \phiPseudoUniqueness \wedge \phiEqualH \wedge
\phiAddOneVBW.$$ 
The definitions of the first and the second part of~$\phivbw$ are simple:
$$
\phiPseudoUniqueness \deff \bigwedge_{c \in \{ \wht, \bl \}} 
\bigwedge_{p \in \{v, h \}} \bigwedge_{i=0}^{n-1} 
\Diamond(c \wedge p_i) \rightarrow \Box(c \wedge p_i)
$$
$$
\phiEqualH \deff \bigwedge_{i=0}^{n-1} 
\Diamond(\bl \wedge h_i) \iff \Diamond(\wht \wedge h_i)
$$

Finally we encode the~$\oplus_{2^n}$-operation as the 
formula~$\phiAddOneVBW$ by, a rather standard, implementation of binary addition.
Below we distinguish two cases: when~$V_b$ is equal to~$2^{n}{-}1$ and 
when~$V_b$ is smaller than~$2^{n}{-}1$.
\begin{align*} 
&\phiAddOneVBW \deff 
( \Diamond(\bl \wedge \bigwedge_{i=0}^{n-1} v_i) \rightarrow 
\Diamond(\wht \wedge \bigwedge_{i=0}^{n-1} \neg v_i)) \; \wedge\\
&\;\;\;\;\;\bigvee_{i=0}^{n-1} ( 
\Diamond(\bl \wedge \neg v_i \wedge \bigwedge_{j=0}^{i-1} v_j) \wedge
\Diamond(\wht \wedge v_i \wedge \bigwedge_{j=0}^{i-1} \neg v_j) \wedge
\bigwedge_{j=i+1}^{n-1} \Diamond(\bl \wedge v_j) \iff \Diamond(\wht \wedge v_j) 
)
\end{align*} 
This completes the definition of~$\phivbw$.  The following three
definitions are analogous.
\begin{align*} 
\phihbw &\deff \phiPseudoUniqueness \wedge \phiEqualV \wedge \phiAddOneHBW\\
\phivwb &\deff \phiPseudoUniqueness \wedge \phiEqualH \wedge \phiAddOneVWB\\
\phihwb &\deff \phiPseudoUniqueness \wedge \phiEqualV \wedge \phiAddOneHWB
\end{align*} 
The formula~$\phiEqualV$ can be obtained from~$\phiEqualH$ by
replacing, for every~$i$, the letter~$h_i$ with the letter~$v_i$, and defining the formulas
$\phiAddOneHBW$,~$\phiAddOneVWB$ and~$\phiAddOneHWB$ as simple
modifications of~$\phiAddOneVBW$. While modifying the mentioned formula
one should only switch~$\bl$ and~$\wht$ propositional symbols and
possibly change~$v$ to~$h$ (when we consider adding~$\oplus_{2^n}1$ on the
$H$ axis).

The following Lemma simply states that the formula~$\phiTorus$ indeed defines a valid torus.
Its proof is routine and follows directly from correctness of all presented formulas.

\begin{lemma} \label{lem:phiTorusCorrect} Assume that the the formula~$\phiTorus$ is locally satisfied at a world~$w$ of a Euclidean
  structure~$\str{A} = \tuple{W,R,V}$. Then, set~$Q_{\str{A}}(w)$,
  i.e., the~$R$-clique for~$w$, contains exactly~$2^n \cdot 2^n$
  elements and each of them carries a different position~$\tuple{H, V}$,
  i.e., there are no two worlds~$v,v'$ satisfying exactly the same~$h_i$- and~$v_i$-predicates.
\end{lemma}
Having defined a proper torus, it is quite easy to encode a solution
to the torus tiling problem~$\cP$ with the initial condition~$c$. Each
inner node will be labelled with a single tile from~$\cT$ and using
appropriate lanterns we enforce that any two neighbouring worlds do not
violate the tiling rules~$\cH$ and~$\cV$. This is the purpose of the formula 
$\phiTiling$ defined below:
$$
\phiTiling \deff \U( \phiTile \wedge \phiInitCond \wedge \phiTilingRules )
$$
The first conjunct specifies that each inner world is labeled 
with exactly one tile. 
$$
\phiTile \deff \inn \rightarrow  (\bigvee_{t \in \cT} t) \wedge 
\bigwedge_{t,t' \in \cT, t \neq t'} (\neg t \vee \neg t')
$$
The second conjunct distributes the initial tiling among torus cells. To define
it we use handy macros~$V{=}k$ and~$H{=}k$, with their intuitive meaning that 
the binary representation of the number~$k$ is written on atomic letters 
$v_0, v_1, \ldots, v_{n-1}$ and~$h_0, h_1, \ldots, h_{n-1}$, respectively. Thus:
$$
\phiInitCond \deff \bigwedge_{i=0}^{n-1} (\inn \wedge H{=}0 \wedge V{=}i) 
\rightarrow t_i
$$
The last formula says that any two successive worlds do not violate tiling rules.
Since any two neighbours are connected via a lantern, we describe the formula from the point
of view of such lantern.

\begin{align*}
\phiTilingRules \deff\; & 
( \lan \wedge \vbw \rightarrow \bigvee_{(t,t') \in \cV} (\Diamond(\bl \wedge t) \wedge \Diamond(\wht \wedge t')) ) \; \wedge \\
& ( \lan \wedge \vwb \rightarrow \bigvee_{(t',t) \in \cV} (\Diamond(\wht \wedge t) \wedge \Diamond(\bl \wedge t')) ) \; \wedge \\
& ( \lan \wedge \hbw \rightarrow \bigvee_{(t,t') \in \cH} (\Diamond(\bl \wedge t) \wedge \Diamond(\wht \wedge t')) ) \; \wedge \\
& ( \lan \wedge \hwb \rightarrow \bigvee_{(t',t) \in \cH} (\Diamond(\wht \wedge t) \wedge \Diamond(\bl \wedge t')) )
\end{align*}

In the following lemma we claim that the presented reduction is correct. Its proof is once again 
routine and follows directly from correctness of all presented formulas.
\begin{lemma}\label{lem:phiTorusCorrectFinal} 
Let~$\phiReduction \deff \phiTorus \wedge \phiTiling$.
The torus tiling problem instance~$\tuple{\cP,c}$ has a solution if and only if 
the formula is~$\phiReduction$ locally-satisfiable.
\end{lemma}

Note that our intended models are serial. Thus, the result holds also
for the logic~$\MLDFive$. This gives the following theorem.

\begin{theorem}\label{thm:theorem5}
The local and global satisfiability problems for the 
logics~$\MLKFive (\Diamond_{\ge}, \RevDiamond)$ 
and~$\MLDFive (\Diamond_{\ge}, \RevDiamond)$ 
are~\NExpTime-hard.
\end{theorem}

Together with Theorem~\ref{thm:ModEuclNExp} this gives:
\begin{theorem} \label{thm:EuclNExp}
The local and global satisfiability problems for the 
logics~$\MLKFive (\Diamond_{\ge}, \RevDiamond)$,~$\MLKFive (\Diamond_{\ge}, \RevDiamond_{\ge})$
and for logics~$\MLDFive (\Diamond_{\ge}, \RevDiamond)$,~$\MLDFive (\Diamond_{\ge}, \RevDiamond_{\ge})$ are \NExpTime-complete.
\end{theorem}

\subsection{Transitive Euclidean frames}
It turns out that the logics of  transitive  Euclidean frames have lower
computational complexity. This is due to the following lemma.

\begin{lemma}
  Let~$\str A$ be an~$R$-connected structure over a transitive Euclidean
  frame~$\tuple{W, R}$. Then, every world~$l \in L_{\str A}$ illuminates~$Q_{\str A}$.\label{lem:PropOfTransEuclStr}
\end{lemma}
\begin{proof}
  Take any world~$q\in Q_{\str A}$. We will show that a lantern~$l$
  illuminates~$q$. Since~$l$ has no~$R$-predecessor and~$\str A$ is
  $R$-connected, there exists a world~$q'\in Q_{\str A}$ such
  that~$\tuple{l, q'} \in R$. By Lemma~\ref{lem:PropOfEuclStr} set
  $Q_{\str A}$ is an~$R$-clique, and thus we
  have~$\tuple{q', q} \in R$. By transitivity we conclude
  that~$\tuple{l, q} \in R$. Thus a lantern~$l$
  illuminates~$Q_{\str A}$.
\end{proof}

A first-order formula stating that all non-lanterns are~$R$-successors
of all lanterns requires only two variables. Thus, as an immediate
conclusion from Lemma~\ref{lem:PropOfTransEuclStr}, we can extend
the translation developed in the previous section to handle the logic
$\MLKFourFive(\Diamond_{\geq}, \RevDiamond_{\geq})$, and obtain 
a \NExpTime{}-upper bound for the satisfiability problem. In fact, the shape
of transitive Euclidean structures is so simple that two-variable
logic is no longer necessary. Below we translate
$\MLKFourFive(\Diamond_{\geq}, \RevDiamond_{\geq})$ and~$\MLDFourFive(\Diamond_{\geq}, \RevDiamond_{\geq})$ to one-variable
logic with counting~\COne, which is \NP-complete~\cite{PHartmann08}.

\begin{theorem}\label{thm:ModTransEuclNP}
  The local and the global satisfiability problems for Transitive Euclidean Modal Logics~$\MLKFourFive(\Diamond_{\geq}, \RevDiamond_{\geq})$ and~$\MLDFourFive(\Diamond_{\geq}, \RevDiamond_{\geq})$ are in~\NP.
\end{theorem}

\begin{proof}

  The proof is similar in spirit to the proof of Lemma~3
  in~\cite{KazakovP09}. Let~$\lan(\cdot)$ be a new unary
  predicate. We first define translation function~$\mathbf{tr}$ that, given
  a~$\MLKFourFive(\Diamond_{\geq}, \RevDiamond_{\geq})$ formula~$\varphi$, produces an equisatisfiable \COne formula~$\tr(\varphi)$. We assume that all counting subscripts~$\varphi$
  are non-zero.
  \begin{gather}
    \tr(p) = p(x) \text{ for all $p \in \Pi$ } \\
    \tr(\varphi \wedge \psi ) = \mathbf{tr}(\varphi) \wedge
    \tr(\psi)  \text{ similarly for~$\neg$,~$\vee$, etc.}\\
    \tr(\Diamond_{\geq n}\varphi) = \exists_{\geq n}.x(\neg\lan(x) \wedge \tr(\varphi))\\
    \tr(\Diamond_{\leq n}\varphi) = \exists_{\leq n}.x(\neg\lan(x) \wedge \tr(\varphi))\\
    \tr(\RevDiamond_{\geq n}\varphi) = \neg\lan(x) \wedge \exists_{\geq n}.x(\tr(\varphi)) \\
    \tr(\RevDiamond_{\leq n}\varphi) = \lan(x) \vee \exists_{\leq n}.x(\tr(\varphi))
  \end{gather}  

  Observe that~$\tr(\varphi)$ is linear in the size of~$\varphi$. 
  Let~$\str B$ be a Kripke structure over a transitive Euclidean frame. 
  Expand~$\str B$ to a structure~$\str B^{+}$ by setting an interpretation of a symbol~$\lan$ to be~$\lan^{\str B^{+}} =\{ w \in \str B \mid w \in L_{\str B}\}$. 
  Taking into account Lemma~\ref{lem:PropOfEuclStr} and
  Lemma~\ref{lem:PropOfTransEuclStr}, a structural induction on~$\varphi$ easily establishes the following condition
 $$ \str B, w_0 \models \varphi\text{ if and only if } \str B^{+} \models 
  \tr(\varphi)[w_0 /x] \text{ for every world~$w_0$}.$$
  Thus, a~$\MLKFourFive(\Diamond_{\geq}, \RevDiamond_{\geq})$ formula~$\varphi$ is locally satisfiable if and only if \COne formula~$\exists_{\geq 1}.{x}(\tr(x))$ is satisfiable, yielding an \NP{}
  algorithm for~$\MLKFourFive(\Diamond_{\geq}, \RevDiamond_{\geq})$ satisfiability. The algorithm for~$\MLDFourFive(\Diamond_{\geq}, \RevDiamond_{\geq})$
is obtained by just a slight update to the one given above. It suffices to supplement the \COne formula defined in the case of~$\MLKFourFive$ with the conjunct~$\exists{x}.(\neg\lan(x) )$ 
  expressing seriality (cf.~the proof of Theorem~\ref{thm:ModEuclNExp}). 
\end{proof}



\section{Transitive frames: counting successors, accessing predecessors} \label{s:transitive}

In this section, we consider the language~$\Diamond_{\ge}, \RevDiamond$, that is the modal language in which we can count the successors,
but cannot count the predecessors, having at our disposal only the basic converse modality. Over all classes of frames involving neither 
transitivity nor Euclideanness local satisfiability is \PSpace-complete and global satisfiability is \ExpTime-complete, as the tight lower and upper bounds can be transferred from, resp., the one-way non-graded language~$\Diamond$ and  the full two-way graded language. Over the classes of Euclidean frames \MLKFive{} and \MLDFive{}, 
both problems are \NExpTime-complete, as proved in Theorem~\ref{thm:EuclNExp}.  Over the classes of transitive Euclidean frames
\MLKBFourFive, \MLKFourFive, \MLDFourFive{}, and  \MLSFive{} the problems are \NP-complete, as the lower bound transfers from the language
$\Diamond$, and the upper bound from the full two-way graded language (Theorem~\ref{thm:ModTransEuclNP}). 
So, over all the above-discussed classes of frames the complexities of~$\Diamond_{\ge}, \RevDiamond$ and~$\Diamond_{\ge}, \RevDiamond_{\ge}$ coincide.

What is left are the classes of transitive frames~$\MLKFour$,~$\MLDFour$, and~$\MLSFour$.
Recall that, in contrast to their one-way counterparts, the two-way graded logics of transitive frames~$\MLKFour (\Diamond_{\ge}, \RevDiamond_{\ge})$, $\MLDFour (\Diamond_{\ge}, \RevDiamond_{\ge})$, and~$\MLSFour (\Diamond_{\ge}, \RevDiamond_{\ge})$  are undecidable~\cite{Zolin17}. 
In \cite{Zolin17} the  question is asked if  the decidability is regained 
when the language is restricted to~$\Diamond_{\ge}, \RevDiamond$. Here we answer this question, demonstrating
the local and global finite model property for the obtained logics; this implies that their satisfiability problems are indeed decidable.

In Lemma~5.5 from~\cite{Zolin17},  it is shown that over the class of transitive frames the global satisfiability and local satisfiability
problems for the considered language are polynomially equivalent. Moreover, they are polynomially equivalent to the \emph{combined} satisfiability 
problem, asking if for a given pair of formulas~$\phi, \phi'$ there exists a structure in which~$\phi$ is true at every world and~$\phi'$ is true at some world.
The remark following the proof of that lemma says that it holds also
for reflexive transitive frames. The same can be easily shown also for serial transitive frames. We thus have:
\begin{lemma}\label{l:locglobequiv}
For each of the logics~$\MLKFour (\Diamond_{\ge}, \RevDiamond)$,~$\MLDFour (\Diamond_{\ge}, \RevDiamond)$, and~$\MLSFour (\Diamond_{\ge}, \RevDiamond)$
their global, local and combined satisfiability problems 
are polynomially equivalent.
\end{lemma}
Below we explicitly deal with global satisfiability. The above lemma implies, however, that  our results apply also to local satisfiability.

Let us concentrate on the class~$\MLKFour$ of all transitive frames. The finite model construction we are going to present is the most complicated part of this paper.
It begins similarly to the exponential model construction in the case of local satisfiability of~$\MLKFour (\Diamond_{\ge})$ 
from~\cite{KazakovP09}: we introduce a Scott-type normal form (Lemma~\ref{l:normalform}), and then generalize  two pieces of model 
surgery used there (Lemma~\ref{l:ian}) to our setting: 
starting from any model, we first obtain a model with short \emph{paths of cliques} and then we decrease the size of the cliques.
Some modifications of the constructions from~\cite{KazakovP09} are necessary to properly deal with the converse modality; they are, however, rather straightforward. Having a model with short paths of cliques and small cliques, we develop some new machinery 
of \emph{clique profiles} and \emph{clique types} allowing us to
decrease the overall size of the structure; this fragment is our main contribution.

\begin{lemma} \label{l:normalform}
Given a formula~$\varphi$ of the language~$(\Diamond_{\ge}, \RevDiamond)$, we can compute in polynomial time a formula~$\psi$ of the form
\begin{align} \label{eq:normalform}
\nonumber \eta \wedge  \bigwedge_{1 \le i \le l} (p_i & \rightarrow \Diamond_{\ge C_i} \pi_i) \wedge 
            \bigwedge_{1 \le i \le m} (q_i \rightarrow \Diamond_{\le D_i} \chi_i) \wedge \\
						& \bigwedge_{1 \le i \le l'} (p'_i \rightarrow \RevDiamond \pi'_i) \wedge
						 \bigwedge_{1 \le i \le m'} (q'_i \rightarrow \boxminus \neg \chi'_i) 
\end{align}
where the~$p_i$,~$q_i$,~$p'_i$,~$q'_i$ are propositional variables, the~$C_i$,~$D_i$  are natural numbers, 
and~$\eta$ and the~$\pi_i$,~$\chi_i$,~$\pi'_i$,~$\chi'_i$ are propositional formulas, such that~$\varphi$ and~$\psi$ are globally  satisfiable
over exactly the same transitive frames.
\end{lemma}

\begin{proof}
  Follows by a routine renaming process, which is similar to the proof
  of Lemma 4 from \cite{KazakovP09}).
\end{proof}

Next, let us introduce some helpful terminology, copying it mostly from the above-mentioned paper \cite{KazakovP09}. Let~$\str{A}=\tuple{W, R, V}$ be a transitive structure, and~$w_1, w_2 \in W$. 
We say that~$w_2$ is an~$R$-\emph{successor} of~$w_1$
if~$ \langle w_1, w_2 \rangle \in R$;~$w_2$ is a \emph{strict}~$R$-\emph{successor} of~$w_1$ if~$ \langle w_1, w_2 \rangle \in R$, 
but~$ \langle w_2, w_1 \rangle \not\in R$;~$w_2$ is a \emph{direct}~$R$-\emph{successor} of~$w_1$ if~$w_2$ is a strict~$R$-successor of~$w_1$ and, for every
$w \in W$ such that~$\langle w_1, w \rangle \in R$ and~$\langle w, w_2 \rangle \in R$ we have either~$w \in Q_\str{A}(w_1)$ or 
$w \in Q_\str{A}(w_2)$. Recall that~$Q_\str{A}(w)$ denotes the~$R$-clique 
for~$w$ in~$\str A$. 

The \emph{depth} of a structure~$\str{A}$ is the maximum over all~$k \ge 0$ for which there exist worlds~$w_0, \ldots, w_k \in W$
such that~$w_i$ is a strict~$R$-successor of~$w_{i-1}$ for every~$1 \le i \le k$, or~$\infty$ if no such a maximum exists. The \emph{breadth}
of~$\str{A}$ is the maximum over all~$k \ge 0$ for which there exist worlds~$w, w_1, \ldots, w_k$ such that~$w_i$ is a direct
$R$-successor of~$w$ for every~$1 \le i \le k$, and the sets~$Q_{\str{A}}(w_1), \ldots Q_{\str{A}}(w_k)$ are disjoint, or~$\infty$ if no
such a maximum exists. The \emph{width} of~$\str{A}$ is the smallest~$k$ such that~$k \ge |Q_\str{A}(w)|$ for all~$w \in W$, or~$\infty$ if
no such~$k$ exists. 

\begin{lemma} \label{l:ian}
Let~$\varphi$ be a normal form formula as in Equation~\ref{eq:normalform}. If~$\varphi$ is globally satisfied in a transitive model~$\str{A}$ then it is globally satisfied
in a transitive model~$\str{A}'$  with depth~$d' \le (\sum_{i=1}^m D_i) +m +m'+1$  and width~$c' \le (\sum_{i=1}^l C_i) + l' + 1$. 
\end{lemma}

\begin{proof}
The proof is a construction being a minor modification of Stages  1 and  4 of the construction from the proof of Lemma 6 in~\cite{KazakovP09}, where the language without backward modalities is considered.
We closely follow the lines of  Kazakov and Pratt-Hartmann's construction,  just taking additional care of backward witnesses.
We remark here that also Stage 2 of the above mentioned construction could be adapted, giving a better bound on the depth 
of~$\str{A}'$. We omit it here since such an improvement would not be crucial for our purposes. Stage 3 cannot be directly adapted. 

Let us turn to the detailed proof.

\medskip\noindent
\emph{Stage 1. Small depth.}
Let~$\str{A}=\tuple{W, R, V}$. For~$w \in W$ define~$d^i_\str{A}(w):=\min(D_i+1, |\{w': \str{A},w' \models \chi_i, \langle w, w' \rangle \in  R^*\}|)$ where~$D_i$ and~$\chi_i$,~$1 \le i \le m$, are as in Equation~\ref{eq:normalform}
and~$R^*$ is the reflexive closure of~$R$. We also define~$S_\str{A}(w):=\{ \chi'_i: \text{there is~$w'$ such that } \str{A}, w' \models \chi'_i \text { and }
\langle w', w \rangle \in R^* \}$, where~$\chi'_i$,~$1 \le i \le m'$ are also as in Equation~\ref{eq:normalform}.  

Let~$R_{\sim}:=\{ \langle w_1, w_2 \rangle \in R: 
d^i_\str{A}(w_1)=d^i_\str{A}(w_2) \text{ for all~$1 \le i \le m$ and } S_\str{A}(w_1)=S_\str{A}(w_2) \}$ be the restriction of~$R$ to pairs of worlds that have the
same values of the~$d^i_\str{A}$ and~$S_\str{A}$.  Let~$R^-_{\sim}$ be the inverse of~$R_{\sim}$. Let~$\str{A}'= \tuple{W, R', V}$ be obtained from~$\str{A}= \tuple{W,R,V}$ by setting~$R':=(R \cup R^-_{\sim})^+$, where the superscript~$+$ is the transitive closure operator. Intuitively, if~$w_1$ is~$R$-reachable from~$w_2$~$w_1$ and~$w_2$
agree on the number (up to the limit of~$D_i$) of the worlds satisfying~$\chi_i$ reachable from them,  for all~$1\le i \le m$, and, for all~$i$,~$w_1$ is 
an~$R$-successor of a world satisfying~$\chi_i'$ iff~$w_2$ is, then we make~$w_1$ and~$w_2$~$R'$-equivalent. The effect is that some~$R$-cliques of~$\str{A}$ are joined into bigger~$R$-cliques in~$\str{A}'$.
We show that~$\str{A}'$ satisfies~$\varphi$ and has appropriately bounded depth.

For every~$w_1, w_2 \in W$ such that~$w_2$ is a strict~$R'$-successor of~$w_1$, we have~$d^i_\str{A}(w_1) \ge d^i_\str{A}(w_2)$ for 
all~$1 \le i \le m$, 
$S_\str{A}(w_1) \subseteq S_\str{A}(w_2)$ and either~$d^i_\str{A}(w_1) > d^i_\str{A}(w_2)$ for some~$i$, and
thus~$\sum_{i=1}^m d^i_\str{A}(w_1) >  \sum_{i=1}^m d^i_\str{A}(w_2)$  or the inclusion
$S_\str{A}(w_1) \subseteq S_\str{A}(w_2)$ is strict. Since~$d^i_\str{A} (w) \le D_i+1$ for every~$w \in W$ and every~$1 \le i \le m$, 
and the size of~$S_\str{A}(w)$ is bounded by~$m'$, the length of every chain~$w_0, \ldots, w_k$ such that~$w_i$ is a strict~$R'$-successor  of~$w_{i-1}$ is bounded by~$(\sum_{i=1}^m D_j) +m +m'+1$.

In order to prove that~$\str{A}' \models \varphi$, we first prove that~$d^i_\str{A}(w) = d^i_{\str{A}'}(w)$ for every~$w \in W$ and~$1 \le i \le m$. Assume to the contrary that~$d^i_\str{A}(w) \not= d^i_{\str{A}'}(w)$ for some~$w \in W$ and some~$i$. 
Since~$R \subseteq R'$, we have~$d^i_\str{A}(w) < d^i_{\str{A}'}(w) \le D_i+1$, which means, in particular, that there exists an element~$w' \in W$ with~$\str{A}, w' \models \chi_i$, such that~$\langle w, w' \rangle \in R'$ but~$\langle w, w' \rangle \not\in R$.

Since~$\langle w, w' \rangle \in R'$, by definition of~$R'$, there exists a sequence~$w_0, \ldots, w_k$ of different worlds in~$W$ such that 
$w_0=w$,~$w_k=w'$ and~$\langle w_{j-1}, w_j \rangle \in R \cup R^-_{\sim}$ for every~$1 \le j \le k$. Note that 
$d^i_\str{A}(w_{j-1}) \ge d^i_\str{A}(w_{j})$ for every~$1 \le j \le k$ and every~$1 \le i \le m$. Take the maximal~$j$ such that~$\langle
w_{j-1}, w' \rangle \not\in R$. Since~$\langle w_0, w' \rangle = \langle w, w' \rangle \not\in R$, such a maximal~$j$ always exists.
Then~$\langle w_j, w' \rangle \in R^*$, and~$\langle w_{j-1}, w_j \rangle \not\in R$. Since~$\langle w_{j-1}, w_j \rangle \in R \cup R^-_{\sim}$, we have~$\langle w_{j-1}, w_j \rangle \in R^-_{\sim}$, and so
$d^i_\str{A}(w_{j-1}) = d^i_{\str{A}}(w_j)$ by definition of~$R_\sim$. Since~$d^i_\str{A}(w_j) \le d^i_\str{A}(w_0) = d^i_\str{A}(w)
< D_i+1$, we obtain a contradiction, due to the fact that~$d^i_{\str{A}}(w_{j-1})=d^i_\str{A}(w_j) \le D_i$,~$\langle w_{j-1}, w\ \rangle
\not\in R^*$,~$\langle w_{j}, w' \rangle \in R^*$,~$\langle w_j, w_{j-1} \rangle \in R$, and~$\str{A}, w' \models \chi_i$.

The observation that~$S_\str{A}(w) = S_{\str{A}'}(w)$ for all~$w \in W$ is even simpler. Assume to the contrary that this 
equality does not hold for some~$w \in W$. This means that~$\chi'_i \in S_{\str{A}'}(w)$ and~$\chi'_i \not\in S_{\str{A}}(w)$ for some~$1 \le i \le m'$. 
In particular, there exists an element~$w' \in W$ with~$\str{A}, w' \models \chi'_i$, such that~$\langle w', w \rangle \in R'$ but~$\langle w', w \rangle \not\in R$.
Thus, there is a 
sequence of different worlds~$w'=w_0, \ldots, w_k=w$ such that 
~$\langle w_{j-1}, w_j \rangle \in R \cup R^-_{\sim}$ for every~$1 \le j \le k$. Note that 
$S_\str{A}(w_{j-1}) \subseteq S_\str{A}(w_{j})$ for every~$1 \le j \le k$. Since~$\chi_i' \in S_\str{A}(w_0)$ it follows
that~$\chi_i' \in S_\str{A}(w_k)$. Contradiction.

To complete the proof that~$\str{A}' \models \varphi$ we demonstrate that, if~$\psi$ is any conjunct of~$\varphi$ and~$w \in W$,
then~$\str{A}, w \models \psi$ implies~$\str{A}', w \models \psi$. Indeed, for the propositional formula~$\eta$ it is immediate.
For subformulas~$(p_i  \rightarrow \Diamond_{\ge C_i} \pi_i)$ and~$(p'_i \rightarrow \RevDiamond \pi'_i)$ this holds since~$R \subseteq R'$.
For subformulas~$(q_i \rightarrow \Diamond_{\le D_i} \chi_i)$ this follows from the property~$d^i_\str{A}(w) = d^i_{\str{A}'}(w)$. Finally,
for subformulas~$(q'_i \rightarrow \boxminus \neg \chi'_i)$ this follows from the property~$S_\str{A}(w)=S_{\str{A}'}(w)$.

\medskip\noindent
\emph{Stage 2. Small width.} By Stage 1 we may assume that~$\str{A}$ has depth bounded by~$(\sum_{i=1}^m D_j) +m +m'+1$.
For every element~$w \in W$ we define~$Q_{\pi_i}(w)$ to be the set of elements of~$Q_\str{A}(w)$ for which~$\pi_i$ holds ($1 \le i \le l$)
and~$Q_{\pi'_i}(w)$ to be the set of elements of~$Q_\str{A}(w)$ for which~$\pi'_i$ holds ($1 \le i \le l'$). We call the elements
of each~$Q_\pi(w)$ the \emph{equivalent}~$\pi$-\emph{witnesses for}~$w$. Note that for each relevant~$\pi$ we have~$Q_\pi(w_1) = Q_\pi(w_2)$ when~$w_1$ and~$w_2$ are
$R$-equivalent. For~$1 \le i \le l$, let~$Q'_{\pi_i}(w)$ be~$Q_{\pi_i}(w)$ if~$|Q_{\pi_i}(w)| \le C_i$, or, otherwise, 
a subset of~$Q_{\pi_i}(w)$ which contains exactly~$C_i$ elements. We call~$Q'_{\pi_i}(w)$ the \emph{selected equivalent~$\pi_i$-witnesses
for~$w$}. For~$1 \le i \le l'$, let~$Q'_{\pi'_i}(w)$ be~$Q_{\pi'_i}(w)$ if~$|Q_{\pi_i}(w)| \le 1$, or, otherwise, 
a singleton subset of~$Q_{\pi_i}(w)$. We call~$Q'_{\pi_i}(w)$ the \emph{selected equivalent~$\pi'_i$-witness for~$w$}.
Additionally, define~$Q'_*(w)$ to be any singleton subset of~$Q_\str{A}(w)$.
We assume that if~$w_1$ and~$w_2$ are~$R$-equivalent then~$Q'_{\pi_i}(w_1) = Q'_{\pi_i}(w_2)$ for all~$1 \le i \le l$,
$Q'_{\pi'_i}(w_1)=Q'_{\pi'_i}(w_2)$ for~$1 \le i \le l'$, and~$Q'_*(w_1)=Q'_*(w_2)$.
Define the structure~$\str{A}'= \tuple{W', R', V'}$ by setting~$W'=\bigcup_{w \in W, 1 \le i \le l} Q'_{\pi_i}(w) \cup \bigcup_{w \in W, 1 \le i \le l'} Q'_{\pi'_i}(w) \cup Q'_*(w)$,~$R':=R \restr W'$, and~$V' = V \restr W'$. Intuitively~$\str{A}'$ is obtained from~$\str{A}$ by removing elements in every~$R$-clique, except for those that are 
selected witnesses for other elements or are members of the singleton set~$Q_*$, guaranteeing that the clique will remain non-empty. It is not 
difficult to see that~$\str{A}'$ has the required properties. In particular our selection process selects 
at most~$(\sum_{i=1}^l C_i) + l' + 1$ elements in every~$R$-clique.
\end{proof}

To describe our next step, we need a few more definitions.
Given a world~$w$ of a structure~$\str{A}$, we define its \emph{depth} as  the maximum over all~$k \ge 0$ for which there exist worlds~$w=w_0, \ldots, w_k \in W$ such that~$w_i$ is a strict~$R$-successor of~$w_{i-1}$ for every $1 \le i \le k$, or as~$\infty$ if no
such a maximum exists.  For an~$R$-clique~$Q$ we define its \emph{depth} as the depth of~$w$ for any~$w \in Q$; this definition is sound since
 for all~$w_1 \in Q_{\str{A}}(w)$ the depth of~$w$ is equal to the depth of~$w_1$.

From this point, we will mostly work on the level of cliques rather than individual worlds. We may view any structure~$\str{A}$
as a partially ordered set of cliques. We write~$\langle Q_1, Q_2 \rangle \in R$, and say that a clique~$Q_1$ \emph{sends} an edge to a clique~$Q_2$ (or that~$Q_2$ \emph{receives}
an edge from~$Q_1$) if~$\langle w_1, w_2 \rangle \in R$ for any (equivalently: for all)~$w_1 \in Q_1$,~$w_2 \in Q_2$. 

A~$1$-type of a world~$w$ in~$\str{A}$ is the set of all  propositional variables
$p$ such that~$\str{A} \models p$. We sometimes identify a~$1$-type with the conjunction of all its elements and negations of variables it does not contain. Given a natural number~$k$, a structure~$\str{A}$ and a clique~$Q$ in this structure~$\str{A}$, we define a  
$k$-\emph{profile} of~$Q$ (called just a \emph{profile} if~$k$ is clear from the context) in~$\str{A}$
 as the tuple~$prof^k_\str{A}(Q)=(\cH, \cA, \cB, \irref)$, where 
$\cH$ is the multiset of~$1$-types in which the number of copies of each~$1$-type~$\alpha$  equals
$\min(k+1, |\{w \in Q: \str{A}, w \models \alpha\}|)$,~$\cA$ is the  multiset of~$1$-types in which the number of copies of each~$1$-type~$\alpha$  equals
$\min(k, |\{w: \str{A}, w \models \alpha \text{ and~$w$ is a strict~$R$-successor  of a world from }~$Q$\}|)$, 
$\cB$ is the  set of~$1$-types of worlds for which a world from~$Q$ is its strict~$R$-successor, and~$\irref$ is a Boolean variable
set to~$1$ iff the clique consists of a single irreflexive element (note that if the clique contains at least two elements then
they all must be reflexive).
Intuitively,~$\cH$ counts (up to~$k+1$)  realizations of~$1$-types \emph{(H)ere} in Q,~$\cA$ counts (up to~$k$) realizations~$1$-types
\emph{(A)bove}~$Q$, and~$\cB$ says which~$1$-types appear \emph{(B)elow}~$Q$.
Usually, given a normal form~$\varphi$ as in Equation~\ref{eq:normalform}, we will be interested in~$M_\varphi$-profiles of cliques,
where~$M_\varphi=\max(\{ C_i \}_{i=1}^{l} \cup \{ D_i + 1\}_{i=1}^{m})$.
Note that, given the~$M_\varphi$-profiles of all cliques in a structure we are able to determine whether this structure is a global model of~$\varphi$.
Indeed, given the~$M_\varphi$-profile of a clique we know the~$1$-types of elements it contains, for each such element 
we can count, at least up to~$M_\varphi$, how many successors of each~$1$-type it has (for this we use the values of~$\cH$,~$\cA$ and~$\irref$), and for each element we know 
the set of~$1$-types of its predecessors (for this we use the values of~$\cH$,~$\cB$ and~$\irref$). Clearly, this information is sufficient to check if every conjunct of~$\varphi$ is satisfied.
The following observation is also straightforward.\\

\begin{lemma} \label{l:profiles}
If~$\str{A} \models \varphi$ for a normal form~$\varphi$, and
if in a structure~$\str{A}'$ the~$M_\varphi$-profile of every 
 clique is equal to the~$M_\varphi$-profile of  some clique from~$\str{A}$, then~$\str{A}' \models \varphi$.
\end{lemma}
 
We now prove the finite model property.

\begin{lemma} \label{l:fmp}
Let~$\varphi$ be a normal form formula. If~$\varphi$ is globally satisfied in a transitive model~$\str{A}$
then it is globally satisfied in a finite transitive model~$\str{A}'$.
\end{lemma}
\begin{proof} 
\emph{Construction of~$\str{A}'$.}
We assume that~$\varphi$ is as in Equation~\ref{eq:normalform}. By Lemma~\ref{l:ian}, we may assume that~$\str{A}=\langle W, R, V \rangle$ has depth~$d \le (\sum_{i=1}^m D_i) +m +m'+1$  and width~$c \le (\sum_{i=1}^l C_i) + l' + 1$. Note that~$\str A$ may be infinite due to possibly infinite breadth.

Let us split~$W$ into sets~$U_0, \ldots, U_d$ with~$U_i$ consisting of all elements of~$W$ of depth~$i$ in~$\str{A}$
(equivalently speaking: being the union of all cliques of depth~$i$ in~$\str{A}$). They are called \emph{layers}. Note that cliques from~$U_i$ may send~$R$-edges
only to cliques from~$U_j$ with~$j<i$.

We now inductively define a sequence of models~$\str{A}=\str{A}_{-1}, \str{A}_0, \ldots, \str{A}_d=\str{A}'$, 
with~$\str{A}_i=\langle W_i, R_i, V_i \rangle$ such that 
\begin{itemize}
\item $W_i= U'_0 \cup \ldots \cup U'_{i} \cup U_{i+1} \cup \ldots \cup U_d$, where each $U'_i$ is a finite union of some cliques from~$U_i$,
\item $V_i = V \restr W_i$
\item $\str{A}_i \restr (U'_0 \cup \ldots \cup U'_{i}) = \str{A}_{i-1} \restr (U'_0 \cup \ldots \cup U'_{i})$,

\item~$\str{A}_i \restr  (U'_0 \cup \ldots \cup U'_{i-1} \cup U_{i+1} \cup \ldots \cup U_d) =
  \str{A}_{i-1} \restr  (U'_0 \cup \ldots \cup U'_{i-1} \cup U_{i+1} \cup \ldots \cup U_d)$  
	\item in particular: 
$\str{A}_i \restr (U_{i+1} \cup \ldots \cup U_d) = \str{A} \restr (U_{i+1} \cup \ldots \cup U_d)$.
\end{itemize}

We obtain~$\str{A}_i$ from~$\str{A}_{i-1}$ by distinguishing a fragment~$U'_i$ of~$U_i$, removing~$U_i \setminus U'_i$ and
adding some edges from~$U_{i+1} \cup \ldots \cup U_{d}$ to~$U'_i$; all the other edges remain untouched. We do it carefully, 
to avoid
modifications of the profiles of
the surviving cliques. Let us describe the process of constructing~$\str{A}_i$ in details.

Assume~$i \ge 0$. We first distinguish a finite subset~$U_i'$ of~$U_i$. We define a \emph{clique type} of every clique~$Q$ from~$U_i$ in~$\str{A}_{i-1}$ as a triple~$(\cH, \cB, S)$,where~$\cH$, and~$\cB$  are as in~$prof^{M_\varphi}_{\str{A}_{i-1}}(Q)$ and~$S$ is the subset of cliques from~$U_0' \cup \ldots \cup U'_{i-1}$, consisting of those cliques
to which~$Q$ sends an~$R_{i-1}$-edge. We stress that during the construction of~$\str{A}_i$, the clique types of cliques
are always computed in~$\str{A}_{i-1}$. In particular~$S$ is empty for~$i=0$ and, as we will always have that 
$U_0' \cup \ldots \cup U'_{i-1}$ is finite,~$S$ is finite for any~$i>0$. Thus for each~$i$ there will be only finitely many clique types. 

For every clique type~$\beta$ realized in~$U_i$, we mark~$M_\varphi$ cliques of this type, or all such
cliques if there are less than~$M_\varphi$ of them. Let~$U'_i$ be the union of the marked cliques.
We fix some arbitrary numbering of the marked cliques.

Now we define the relation~$R_i$. As said before, for any pair of cliques~$Q_1, Q_2$ both of which are contained in
$U'_0 \cup \ldots \cup U'_{i-1} \cup U_{i+1} \cup \ldots \cup U_d$ or in 
$U'_0 \cup \ldots \cup U'_{i}$, we set~$\langle Q_1, Q_2 \rangle \in R_i$ iff
$\langle Q_1, Q_2 \rangle \in R_{i-1}$. It remains to define  
 the~$R_i$-edges from~$U_{i+1} \cup \ldots \cup U_{d}$ to~$U'_i$.
For every clique~$Q$  from~$U_{i+1} \cup \ldots \cup U_d$ and every clique type~$\beta$ realized in~$U_i'$, 
let~$f(\beta)$ be the number of~$R_{i-1}$-edges sent  by~$Q$ to cliques of type~$\beta$ in~$U_i$, if this number is not greater than~$M_\varphi$,
or, otherwise, let~$f(\beta)=M_\varphi$. Let~$f'(\beta)$ be the number of~$R_{i-1}$-edges sent by~$Q$ to cliques of type~$\beta$ in~$U'_i$ (recall that this number is
not greater than~$M_\varphi$). 
We let all the~$R_{i-1}$-edges sent by $Q$ to the cliques of type~$\beta$ in~$U'_i$ to be also members of $R_i$, that is, to be edges in $\str{A}_i$.
Additionally, we link $Q$ by $R_i$ to the first (with respect to the numbering we have fixed) $f(\beta)-f'(\beta)$ cliques of type~$\beta$ in $U'_i$ 
to which $Q$ is not linked by $R_{i-1}$.
By the choice of~$U'_i$, we have enough such cliques in~$U'_i$. 
%
We finish the construction of~$\str{A}_{i}$
by removing all the cliques from~$U_i \setminus U'_i$. 

That~$\str{A}_i$ has the desired properties is shown in the following two claims.

\medskip\noindent
\emph{Claim 1: Each of the~$\str{A}_i$ is a transitive structure.}

We show this by induction by~$i={-1}, 0, \ldots, d$. Obviously~$\str{A}_{-1}=\str{A}$ is transitive. Assume that~$\str{A}_{i-1}$ is transitive,
and assume to the contrary that~$\str{A}_i$ is not. This means there are cliques~$Q_1, Q_2, Q_3$ in~$\str{A}_i$ such that
$\langle Q_1, Q_2 \rangle \in R_i$,~$\langle Q_2, Q_3 \rangle \in R_i$ but~$\langle Q_1, Q_3 \rangle \not\in R_i$. 
It is easy to see that the cliques~$Q_1, Q_2, Q_3$ must belong to three different  layers, and that precisely
one of the two cases holds: either~$Q_2 \subseteq U'_i$,~$\langle Q_1, Q_2 \rangle \not\in R_{i-1}$,~$\langle Q_2, Q_3 \rangle \in R_{i-1}$ 
or~$Q_3 \subseteq U'_i$,~$\langle Q_1, Q_2 \rangle \in R_{i-1}$,~$\langle Q_2, Q_3 \rangle \not\in R_{i-1}$.
In the first case, our construction implies that there is a clique~$Q' \subseteq U_i \setminus U'_i$ such
that~$\langle Q_1, Q' \rangle \in R_{i-1}$, and the clique-types of~$Q_2$ and~$Q'$ are identical. But from the latter
it follows that~$\langle Q', Q_3 \rangle \in R_{i-1}$ and from transitivity of~$R_{i-1}$ we have~$\langle Q_1, Q_3 \rangle \in R_{i-1}$.
Since none of~$Q_1$,~$Q_3$ is  contained in~$U_i$, by our construction we have that~$\langle Q_1, Q_3 \rangle \in R_{i}$.
Contradiction. In the second case, let~$\beta$ be the clique-type of~$Q_3$ and let~$Q'_1, \ldots, Q'_{k_1}$ be the cliques of type~$\beta$ from~$U'_i$ to which~$Q_2$ sends~$R_{i-1}$-edges,
~$Q''_1, \ldots, Q''_{k_2}$ be the cliques of type~$\beta$ from~$U_i \setminus U'_i$ to which~$Q_2$ sends~$R_{i-1}$-edges, and let
~$Q'''_1, \ldots, Q'''_{k_3}$ be the cliques of type~$\beta$ from~$U'_i$ to which~$Q_1$ sends~$R_{i-1}$-edges, but~$Q_2$ does not.
Note that, by transitivity of~$R_{i-1}$,~$Q_1$ sends~$R_{i-1}$-edges to all of the~$Q'_i$ and all of the~$Q''_i$. 
If~$k_1+k_2+k_3 \ge M_\varphi$ then~$Q_1$ must send, by our construction, an~$R_i$-edge to every clique of type~$\beta$ from~$U'_i$, in particular
to~$Q_3$; contradiction. Thus~$k_1+k_2+k_3 < M_\varphi$ and~$Q_1$ sends at least~$k_2$~$R_i$-edges to cliques of type~$\beta$ from~$U'_i$ to which
it does not sent~$R_{i-1}$-edges.~$Q_2$ sends precisely~$k_2$ such edges. Thus, since our strategy of choosing always cliques of type~$\beta$ with
minimal possible numbers in the numbering we have fixed requires~$Q_2$ to send an~$R_i$-edge to~$Q_3$, the same strategy requires~$Q_1$ also to to send  
an~$R_i$-edge to~$Q_3$. Contradiction.

\medskip\noindent
\emph{Claim 2: The~$M_\varphi$-profile of every clique in~$\str{A}_i$ is the same as its~$M_\varphi$-profile in~$\str{A}$.}
Again we work by induction. Assume that the~$M_\varphi$-profiles of the surviving cliques in~$\str{A}_{i-1}$ are the same as in~$\str{A}$. 
We show that the~$M_\varphi$-profiles of cliques surviving in~$\str{A}_{i}$ are the same as in~$\str{A}_{i-1}$. It is obvious for
the~$\cH$-components and the values of~$\irref$, as we do not change the cliques.
The~$\cA$-components for the cliques from~$U'_0 \ldots U'_i$ cannot change since they send~$R_i$-edges to precisely
the same cliques they send~$R_{i-1}$-edges. 
Similarly, the~$\cB$-components for the cliques from~$U_{i+1} \ldots U_d$ cannot change since they receive~$R_i$-edges  precisely from the same cliques they receive~$R_{i-1}$-edges. 

Consider a clique~$Q$ from~$U'_0 \ldots U'_{i-1}$. Note that~$prof^{M_\varphi}_{\str{A}_{i}}(Q).\cB \subseteq prof^{M_\varphi}_{\str{A}_{i-1}}(Q).\cB$ since
any~$R_i$-edge received by~$Q$ is also an~$R_{i-1}$-edge. To see that~$\supseteq$ also holds take any 1-type~$\alpha \in prof^{M_\varphi}_{\str{A}_{i-1}}(Q).\cB$.
Then there exists a clique~$Q'$ containing a realization of~$\alpha$ such that~$Q'$ sends an~$R_{i-1}$-edge to~$Q$. If~$Q'$ survives
in~$\str{A}_i$ then it sends an~$R_i$-edge to~$Q$. Otherwise~$Q' \subseteq U_i \setminus U'_i$ and there is a clique~$Q''$ of the same
clique-type as~$Q'$ in~$U'_i$. This equality of the clique-types implies that~$\alpha$ is realized in~$Q''$ and~$Q''$ sends an
$R_i$-edge to~$Q$. It follows that~$\alpha \in prof^{M_\varphi}_{\str{A}_{i}}(Q)$. Thus~$prof^{M_\varphi}_{\str{A}_{i-1}}(Q)=prof^{M_\varphi}_{\str{A}_{i}}(Q)$.

Consider a clique~$Q$ from~$U'_i$. Obviously
$prof^{M_\varphi}_{\str{A}_{i-1}}(Q).\cB \subseteq
prof^{M_\varphi}_{\str{A}_{i}}(Q).\cB$ since all the~$R_{i-1}$-edges
received by~$Q$ remain~$R_i$-edges.  To see~$\supseteq$ assume
$\alpha \in prof^{M_\varphi}_{\str{A}_{i}}(Q).\cB$ for some~$1$-type
$\alpha$. Then there exists a clique~$Q'$ containing a realization of
$\alpha$ such that~$Q'$ sends an~$R_i$-edge to~$Q$. If~$Q'$ sends also
an~$R_{i-1}$-edge to~$Q$ then
$\alpha \in prof^{M_\varphi}_{\str{A}_{i-1}}(Q).\cB$.  Otherwise, by
our construction,~$Q'$ sends an~$R_{i-1}$-edge to a clique
$Q'' \subseteq U_i \setminus U'_i$ such that the clique-types of~$Q$
and~$Q''$ are equal. But then~$\alpha$ belongs to the~$\cB$-component
of the clique-type of~$Q''$ and also of~$Q$.  So,
$\alpha \in prof^{M_\varphi}_{\str{A}_{i-1}}(Q).\cB$. It follows that
$prof^{M_\varphi}_{\str{A}_{i-1}}(Q).\cB=prof^{M_\varphi}_{\str{A}_{i}}(Q).\cB$.

Finally, consider a clique~$Q$ from~$U_{i+1} \cup \ldots \cup U_d$. It remains to show that 
$prof^{M_\varphi}_{\str{A}_{i-1}}(Q).\cA=prof^{M_\varphi}_{\str{A}_{i}}(Q).\cA$. By our construction, the~$R_i$-edges sent by~$Q$ to~$U'_0 \cup \ldots \cup
U'_{i-1} \cup U_{i+1} \cup \ldots \cup U_d$ are the same as~$R_{i-1}$-edges sent by~$Q$ to this set. The desired equality of the~$\cA$-components (as multisets) follows now easily from the fact that, for any clique-type~$\beta$, whenever~$Q$ sends precisely~$k$~$R_{i-1}$-edges to cliques of~$U_i$ of type~$\beta$ then 
it sends precisely~$k'$-edges to cliques of~$U'_i$ of type~$\beta$, where~$k'=min(k, M_\varphi)$. 
This finishes the proof of Claim 2. 

The two above claims and  Lemma~\ref{l:profiles} imply that~$\str{A}' = \str{A}_d$ is indeed a model of~$\varphi$. As each of the
$U_i'$ contains a finite number of cliques and each of the cliques is finite, we get that~$\str{A}'$ is finite. 
This finishes the proof of Lemma~\ref{l:fmp}.
\end{proof}

Let us estimate the size of the constructed finite model~$\str{A}'$. For~$U'_0$ we take at most~$M_\varphi$ realizations of every clique type from~$U_0$. 
$M_\varphi$ is bounded exponentially, and the number of possible clique types in~$U_0$ is bounded doubly exponentially in~$|\varphi|$
(note that such cliques do not send any edges). Then, to construct~$U'_i$ we consider clique types distinguished, in particular,
by the sets of cliques from~$U'_0 \cup \ldots U'_{i-1}$ to which a given clique sends edges. Thus, the number of cliques in~$U'_i$ may become
exponentially larger than the number of cliques in~$U'_{i-1}$. Thus, we can only estimate the number of cliques in our eventual
finite model by a tower of exponents of height~$d$ (recall that our  bound on~$d$ is exponential in~$|\varphi|$, though a polynomial bound would not be difficult to obtain). We leave open the question if a construction building smaller (e.g., doubly exponential in~$|\varphi|$) models exist.

A careful inspection shows that all our constructions respect  reflexivity and seriality, that is, if we replace the word \emph{transitive}
in the statements of Lemma~\ref{l:ian} and Lemma~\ref{l:fmp} with the phrases \emph{reflexive transitive} or \emph{serial transitive} then they
remain correct.
\begin{theorem}
The logics~$\MLKFour (\Diamond_{\ge}, \RevDiamond)$,~$\MLDFour (\Diamond_{\ge}, \RevDiamond)$,~$\MLSFour (\Diamond_{\ge}, \RevDiamond)$ have the finite model property. Their  local and global satisfiability problems 
are decidable.
\end{theorem}

A natural decision procedure arising from our work is as follows: guess a finite model of the given formula and 
check that it indeed is a model. However, this procedure does not give a good upper complexity bound, since it needs to take into account very large finite models.
The precise complexity can be established using the above-mentioned results from~\cite{GGI19} concerning the description logic~$\mathcal{SIQ}^-$.\\

\begin{theorem}
The local and global satisfiability problems for the logics 
$\MLKFour (\Diamond_{\ge}, \RevDiamond)$,~$\MLDFour (\Diamond_{\ge}, \RevDiamond)$,~$\MLSFour (\Diamond_{\ge}, \RevDiamond)$ are
\TwoExpTime-complete.
\end{theorem}
\begin{proof}
In~\cite{GGI19} it is shown that the knowledge base satisfiability problem for the 
logic~$\mathcal{SIQ}^-$, restricted to a single transitive role,
 is \TwoExpTime-complete.
With this single role restriction, the language of~$\mathcal{SIQ}^-$  becomes a syntactic variant of \MLKFour~$(\Diamond_{\ge}, \RevDiamond)$. 
The knowledge base satisfiability in~$\mathcal{SIQ}^-$ is the question if for a given 
pair~$(\cT, \cA)$, where~$\cT$ is a
\emph{TBox} and~$\cA$ is an \emph{ABox},
there exists a structure containing~$\cA$ and respecting~$\cT$ at every element.
 We do not want to define these notions formally here
and refer the interested reader to~\cite{GGI19} or some other articles on description logics. For our purposes it is sufficient to say that 
$\cT$ consists of implications of the form~$\phi \rightarrow \psi$ and~$\cA$ is a collection of assertions of the form 
$\phi(a)$ or~$T(a,b)$ where~$a, b$ are names for domain elements (which can be used only in~$\cA$),~$\phi(a)$ means that~$\phi$ is satisfied at~$a$,
and~$T(a,b)$ means that there is an edge from~$a$ to~$b$.

To solve  global satisfiability for~$\MLKFour (\Diamond_{\ge}, \RevDiamond)$ we
just translate the input formula~$\phi$ to the knowledge base~$(\{ \top \rightarrow \phi \}, \emptyset )$ and ask for its satisfiability.
Regarding the lower bound, we can easily adapt the lower bound proof from~\cite{GGI19} (Theorem 4) to our scenario.
The proof there goes by a reduction from the acceptance problem for alternating Turing machines with exponentially bounded space, and
uses both  TBoxes and  ABoxes. However, ABoxes are always  of a simple form~$\phi'(a)$. What we can do
is to take the conjunction~$\phi$ of the~$\MLKFour$-counterparts of the implications from the given TBox 
and ask for combined satisfiability of~$\phi$ and~$\phi'$. 
This gives the~$\TwoExpTime$-lower bound for the combined complexity of~$\MLKFour (\Diamond_{\ge}, \RevDiamond)$.
Due to Lemma~\ref{l:locglobequiv} we infer \TwoExpTime-completeness 
of local and global satisfiability in~$\MLKFour (\Diamond_{\ge}, \RevDiamond)$. 

The upper and lower complexity bounds for~$\MLKFour (\Diamond_{\ge}, \RevDiamond)$ and~$\MLSFour (\Diamond_{\ge}, \RevDiamond)$ 
can be obtained by an inspection of the proofs from~\cite{GGI19} and observing that they
work for structures with a reflexive or serial transitive relation.
\end{proof}



\section{Missing lower bounds for logics with converse and without graded modalities} \label{s:nongraded}

To complete the picture we consider in this section the modal language with converse but without graded modalities. 
Over most relevant classes of frames tight complexity bounds for local and global satisfiability of this language are known. However, according to Zolin's survey~\cite{Zolin17}, the three logics of transitive frames \MLKFour~$(\Diamond, \RevDiamond)$, \MLSFour~$(\Diamond, \RevDiamond)$ and \MLDFour~$(\Diamond, \RevDiamond)$ whose global satisfiability is
known to be in \ExpTime{} lack the corresponding lower bound. 
We provide it here. 
We were also not able to find a tight lower bound in the literature for the logics of Euclidean frames, \MLKFive~$(\Diamond, \RevDiamond)$,
\MLDFive~$(\Diamond, \RevDiamond)$.
We also show it here. 
Interestingly, the two reductions are identical, i.e., in both cases we produce the same formulas 
(but the shapes of the intended models differ).

In the conference version of this paper we used a rather heavy reductions from the halting problem for alternating Turing machines working in polynomial space. 
Following the suggestion of one of the referees we looked for an alternative proof by a reduction from global satisfiability of the logic $\MLK(\Diamond)$.
The general idea is  essentially the same as in our previous proof but the reduction is arguably simpler.

 \begin{theorem} \label{t:expharda} 
 The global satisfiability problem for~$\MLKFour(\Diamond, \RevDiamond)$,~$\MLDFour(\Diamond, \RevDiamond)$
   and~$\MLSFour(\Diamond, \RevDiamond)$ is~\ExpTime-hard.
\end{theorem}

\begin{proof}
We recall that global satisfiability problem for $\MLK(\Diamond)$ is \ExpTime-hard. We reduce this problem simultaneously to
global satisfiability of the three logics we consider.
 
Take any modal formula $\varphi$ of $\MLK(\Diamond)$. Without loss of generality we assume that $\varphi$ contains no nested occurrences of $\Diamond$ and $\Box$.
(Indeed, if $\varphi$ contains a nested occurrence of a modal operator, that is it contains a subformula $\Diamond \psi$ or $\Box \psi$ in the scope of another
$\Diamond$ or $\Box$, then we replace that subformula by a
fresh variable $p$ and append the conjunct~$p \leftrightarrow \Diamond \psi$, resp., $p \leftrightarrow \Diamond \psi$. Successively treating in this way all
occurrences of modal operators we eventually end up with a formula equisatisfiable to $\varphi$ in which they are not nested.)

Assuming that $c_0, c_1, c_2$ and~$c_3$ are fresh propositional variables not occurring in $\varphi$ we define the translation $\tr(\varphi)$ as~follows:
\begin{itemize}
\item $\tr(p) = p$ for all propositional variables $p$, 
\item $\tr(\varphi' \lor \varphi'') = \tr(\varphi') \lor \tr(\varphi'')$ and analogously for $\lor, \rightarrow, \leftrightarrow$,
\item $\tr(\neg \varphi') = \neg \tr(\varphi')$,
\item $\tr(\Diamond\varphi') = [c_0 \rightarrow \Diamond(c_1 \land \tr(\varphi'))] \land [c_1 \rightarrow \RevDiamond(c_2 \land \tr(\varphi'))] \land [c_2 \rightarrow \Diamond(c_3 \land \tr(\varphi'))] \land [c_3 \rightarrow \RevDiamond(c_0 \land \tr(\varphi'))]$
and analogously for $\Box\varphi'$
\end{itemize}

Let $\varphi^* = \tr(\varphi) \land (\bigvee_{0 \leq i \leq 3} c_i) \land (\bigwedge_{0 \leq i < j \leq 3} (\neg c_i \vee \neg c_j))$. 
Note that $\varphi^*$ is composed of the translated $\varphi$ and a formula stipulating that for each node exactly one of $c_i$ holds true.
The size of $\varphi^*$ is clearly polynomial in $|\varphi|$ since  $\Diamond$ and $\Box$ have no nested occurrences in~$\varphi$.

\begin{cla}
If $\varphi$ is globally satisfiable in $\MLK(\Diamond)$ then $\varphi^*$ is globally satisfiable in $\MLKFour(\Diamond, \RevDiamond)$,~$\MLDFour(\Diamond, \RevDiamond)$ and~$\MLSFour(\Diamond, \RevDiamond)$.
\end{cla}
\begin{proof}
Let $\str{A}= \langle W, R, V \rangle$ be a model of $\varphi$. 
We assume that $\str{A}$ is tree-shaped 
(this is done without loss of generality since $\MLK(\Diamond)$ has the tree-shaped model property). 
Let $w_r$ denotes the root of $\str{A}$.

For any world $w$ define its distance from the root, denoted with $d(w)$, as the length of the $R$-path from $w_r$, i.e., $d(w_r) = 0$, $d(w) = 1$ iff $R(w_r, w)$ holds, $d(w) = 2$ iff there is a world $v$ such that $R(w_r,v), R(v,w)$ etc.
We define the Kripke structure $\str{A}' = (W', R', V')$ by inverting every second $R$-edge of $\str{A}$ and labelling the worlds on every path, leading from the root, repetitively $c_0, c_1, c_2, c_3, c_0, \ldots$. Formally:
\begin{itemize}
    \item $W = W'$,
    \item For every propositional variable $p \not\in \{ c_0, c_1, c_2, c_3 \}$ we set $V'(p) = V(p)$
    while for the variables $c_i$ we set $V'(c_i) = \{ w : d(w) \; \mathop{mod} \; 4 = i \}$ for $i \in \{ 0,1,2,3 \}$,
    \item $R'$ is the reflexive closure of $R_{(0,1)} \cup R_{(1,2)}^{-1} \cup R_{(2,3)} \cup R_{(3, 0)}^{-1} $, with $R_{(i,j)} = R \cap V'(c_i) \times V'(c_j)$.
\end{itemize}
The shape of the obtained model is illustrated in Fig.~\ref{fig:trees}.

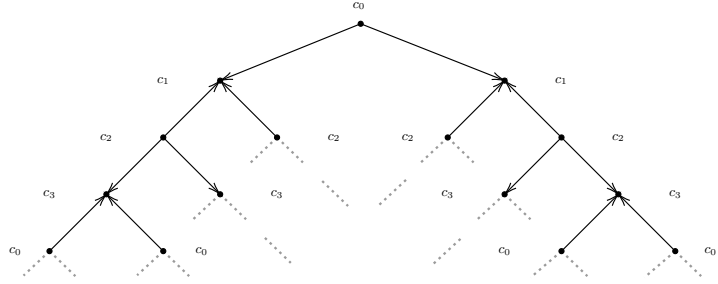
\begin{figure}[h!]   
\begin{center}
\resizebox{0.72\linewidth}{!}{

\tikzset{every picture/.style={line width=0.75pt}} 

\begin{tikzpicture}[x=0.75pt,y=0.75pt,yscale=-1,xscale=1]

\draw [color={rgb, 255:red, 155; green, 155; blue, 155 }  ,draw opacity=1 ][line width=1.5]  [dash pattern={on 1.69pt off 2.76pt}]  (600,250) -- (572.91,273.14) ;

\draw [color={rgb, 255:red, 155; green, 155; blue, 155 }  ,draw opacity=1 ][line width=1.5]  [dash pattern={on 1.69pt off 2.76pt}]  (500,250) -- (522.75,272) ;

\draw [color={rgb, 255:red, 155; green, 155; blue, 155 }  ,draw opacity=1 ][line width=1.5]  [dash pattern={on 1.69pt off 2.76pt}]  (50,250) -- (72.75,272) ;

\draw [color={rgb, 255:red, 155; green, 155; blue, 155 }  ,draw opacity=1 ][line width=1.5]  [dash pattern={on 1.69pt off 2.76pt}]  (450,200) -- (472.75,222) ;

\draw [color={rgb, 255:red, 155; green, 155; blue, 155 }  ,draw opacity=1 ][line width=1.5]  [dash pattern={on 1.69pt off 2.76pt}]  (400,150) -- (422.75,172) ;

\draw [color={rgb, 255:red, 155; green, 155; blue, 155 }  ,draw opacity=1 ][line width=1.5]  [dash pattern={on 1.69pt off 2.76pt}]  (250,150) -- (224.21,174.54) ;

\draw [color={rgb, 255:red, 155; green, 155; blue, 155 }  ,draw opacity=1 ][line width=1.5]  [dash pattern={on 1.69pt off 2.76pt}]  (200,200) -- (175.58,222.58) ;

\draw [color={rgb, 255:red, 155; green, 155; blue, 155 }  ,draw opacity=1 ][line width=1.5]  [dash pattern={on 1.69pt off 2.76pt}]  (150,250) -- (127,271.2) ;

\draw [color={rgb, 255:red, 155; green, 155; blue, 155 }  ,draw opacity=1 ][line width=1.5]  [dash pattern={on 1.69pt off 2.76pt}]  (50,250) -- (27.17,271.5) ;

\draw [color={rgb, 255:red, 155; green, 155; blue, 155 }  ,draw opacity=1 ][line width=1.5]  [dash pattern={on 1.69pt off 2.76pt}]  (600,250) -- (622.75,272) ;

\draw [color={rgb, 255:red, 155; green, 155; blue, 155 }  ,draw opacity=1 ][line width=1.5]  [dash pattern={on 1.69pt off 2.76pt}]  (500,250) -- (472.91,273.14) ;

\draw [color={rgb, 255:red, 155; green, 155; blue, 155 }  ,draw opacity=1 ][line width=1.5]  [dash pattern={on 1.69pt off 2.76pt}]  (450,200) -- (425.54,222.5) ;

\draw [color={rgb, 255:red, 155; green, 155; blue, 155 }  ,draw opacity=1 ][line width=1.5]  [dash pattern={on 1.69pt off 2.76pt}]  (400,150) -- (376.21,172.17) ;

\draw [color={rgb, 255:red, 155; green, 155; blue, 155 }  ,draw opacity=1 ][line width=1.5]  [dash pattern={on 1.69pt off 2.76pt}]  (250,150) -- (272.75,172) ;

\draw [color={rgb, 255:red, 155; green, 155; blue, 155 }  ,draw opacity=1 ][line width=1.5]  [dash pattern={on 1.69pt off 2.76pt}]  (200,200) -- (222.75,222) ;

\draw [color={rgb, 255:red, 155; green, 155; blue, 155 }  ,draw opacity=1 ][line width=1.5]  [dash pattern={on 1.69pt off 2.76pt}]  (150,250) -- (172.75,272) ;

\draw    (451.41,101.41) -- (500,150) ;

\draw [shift={(450,100)}, rotate = 45] [color={rgb, 255:red, 0; green, 0; blue, 0 }  ][line width=0.75]    (10.93,-3.29) .. controls (6.95,-1.4) and (3.31,-0.3) .. (0,0) .. controls (3.31,0.3) and (6.95,1.4) .. (10.93,3.29)   ;
\draw    (323.5,50) -- (448.14,99.26) ;
\draw [shift={(450,100)}, rotate = 201.57] [color={rgb, 255:red, 0; green, 0; blue, 0 }  ][line width=0.75]    (10.93,-3.29) .. controls (6.95,-1.4) and (3.31,-0.3) .. (0,0) .. controls (3.31,0.3) and (6.95,1.4) .. (10.93,3.29)   ;

\draw    (500,150) -- (548.59,198.59) ;
\draw [shift={(550,200)}, rotate = 225] [color={rgb, 255:red, 0; green, 0; blue, 0 }  ][line width=0.75]    (10.93,-3.29) .. controls (6.95,-1.4) and (3.31,-0.3) .. (0,0) .. controls (3.31,0.3) and (6.95,1.4) .. (10.93,3.29)   ;

\draw    (551.41,201.41) -- (600,250) ;

\draw [shift={(550,200)}, rotate = 45] [color={rgb, 255:red, 0; green, 0; blue, 0 }  ][line width=0.75]    (10.93,-3.29) .. controls (6.95,-1.4) and (3.31,-0.3) .. (0,0) .. controls (3.31,0.3) and (6.95,1.4) .. (10.93,3.29)   ;
\draw    (448.59,101.41) -- (400,150) ;

\draw [shift={(450,100)}, rotate = 135] [color={rgb, 255:red, 0; green, 0; blue, 0 }  ][line width=0.75]    (10.93,-3.29) .. controls (6.95,-1.4) and (3.31,-0.3) .. (0,0) .. controls (3.31,0.3) and (6.95,1.4) .. (10.93,3.29)   ;
\draw    (500,150) -- (451.41,198.59) ;
\draw [shift={(450,200)}, rotate = 315] [color={rgb, 255:red, 0; green, 0; blue, 0 }  ][line width=0.75]    (10.93,-3.29) .. controls (6.95,-1.4) and (3.31,-0.3) .. (0,0) .. controls (3.31,0.3) and (6.95,1.4) .. (10.93,3.29)   ;

\draw    (548.59,201.41) -- (500,250) ;

\draw [shift={(550,200)}, rotate = 135] [color={rgb, 255:red, 0; green, 0; blue, 0 }  ][line width=0.75]    (10.93,-3.29) .. controls (6.95,-1.4) and (3.31,-0.3) .. (0,0) .. controls (3.31,0.3) and (6.95,1.4) .. (10.93,3.29)   ;
\draw  [color={rgb, 255:red, 0; green, 0; blue, 0 }  ,draw opacity=1 ][fill={rgb, 255:red, 0; green, 0; blue, 0 }  ,fill opacity=1 ] (447.79,100) .. controls (447.79,98.78) and (448.78,97.79) .. (450,97.79) .. controls (451.22,97.79) and (452.21,98.78) .. (452.21,100) .. controls (452.21,101.22) and (451.22,102.21) .. (450,102.21) .. controls (448.78,102.21) and (447.79,101.22) .. (447.79,100) -- cycle ;
\draw  [color={rgb, 255:red, 0; green, 0; blue, 0 }  ,draw opacity=1 ][fill={rgb, 255:red, 0; green, 0; blue, 0 }  ,fill opacity=1 ] (497.79,150) .. controls (497.79,148.78) and (498.78,147.79) .. (500,147.79) .. controls (501.22,147.79) and (502.21,148.78) .. (502.21,150) .. controls (502.21,151.22) and (501.22,152.21) .. (500,152.21) .. controls (498.78,152.21) and (497.79,151.22) .. (497.79,150) -- cycle ;
\draw  [color={rgb, 255:red, 0; green, 0; blue, 0 }  ,draw opacity=1 ][fill={rgb, 255:red, 0; green, 0; blue, 0 }  ,fill opacity=1 ] (397.79,150) .. controls (397.79,148.78) and (398.78,147.79) .. (400,147.79) .. controls (401.22,147.79) and (402.21,148.78) .. (402.21,150) .. controls (402.21,151.22) and (401.22,152.21) .. (400,152.21) .. controls (398.78,152.21) and (397.79,151.22) .. (397.79,150) -- cycle ;
\draw  [color={rgb, 255:red, 0; green, 0; blue, 0 }  ,draw opacity=1 ][fill={rgb, 255:red, 0; green, 0; blue, 0 }  ,fill opacity=1 ] (447.79,200) .. controls (447.79,198.78) and (448.78,197.79) .. (450,197.79) .. controls (451.22,197.79) and (452.21,198.78) .. (452.21,200) .. controls (452.21,201.22) and (451.22,202.21) .. (450,202.21) .. controls (448.78,202.21) and (447.79,201.22) .. (447.79,200) -- cycle ;
\draw  [color={rgb, 255:red, 0; green, 0; blue, 0 }  ,draw opacity=1 ][fill={rgb, 255:red, 0; green, 0; blue, 0 }  ,fill opacity=1 ] (547.79,200) .. controls (547.79,198.78) and (548.78,197.79) .. (550,197.79) .. controls (551.22,197.79) and (552.21,198.78) .. (552.21,200) .. controls (552.21,201.22) and (551.22,202.21) .. (550,202.21) .. controls (548.78,202.21) and (547.79,201.22) .. (547.79,200) -- cycle ;
\draw  [color={rgb, 255:red, 0; green, 0; blue, 0 }  ,draw opacity=1 ][fill={rgb, 255:red, 0; green, 0; blue, 0 }  ,fill opacity=1 ] (497.79,250) .. controls (497.79,248.78) and (498.78,247.79) .. (500,247.79) .. controls (501.22,247.79) and (502.21,248.78) .. (502.21,250) .. controls (502.21,251.22) and (501.22,252.21) .. (500,252.21) .. controls (498.78,252.21) and (497.79,251.22) .. (497.79,250) -- cycle ;
\draw  [color={rgb, 255:red, 0; green, 0; blue, 0 }  ,draw opacity=1 ][fill={rgb, 255:red, 0; green, 0; blue, 0 }  ,fill opacity=1 ] (597.79,250) .. controls (597.79,248.78) and (598.78,247.79) .. (600,247.79) .. controls (601.22,247.79) and (602.21,248.78) .. (602.21,250) .. controls (602.21,251.22) and (601.22,252.21) .. (600,252.21) .. controls (598.78,252.21) and (597.79,251.22) .. (597.79,250) -- cycle ;
\draw    (323.5,50) -- (201.85,99.25) ;
\draw [shift={(200,100)}, rotate = 337.96000000000004] [color={rgb, 255:red, 0; green, 0; blue, 0 }  ][line width=0.75]    (10.93,-3.29) .. controls (6.95,-1.4) and (3.31,-0.3) .. (0,0) .. controls (3.31,0.3) and (6.95,1.4) .. (10.93,3.29)   ;

\draw    (198.59,101.41) -- (150,150) ;

\draw [shift={(200,100)}, rotate = 135] [color={rgb, 255:red, 0; green, 0; blue, 0 }  ][line width=0.75]    (10.93,-3.29) .. controls (6.95,-1.4) and (3.31,-0.3) .. (0,0) .. controls (3.31,0.3) and (6.95,1.4) .. (10.93,3.29)   ;
\draw    (150,150) -- (101.41,198.59) ;
\draw [shift={(100,200)}, rotate = 315] [color={rgb, 255:red, 0; green, 0; blue, 0 }  ][line width=0.75]    (10.93,-3.29) .. controls (6.95,-1.4) and (3.31,-0.3) .. (0,0) .. controls (3.31,0.3) and (6.95,1.4) .. (10.93,3.29)   ;

\draw    (98.59,201.41) -- (50,250) ;

\draw [shift={(100,200)}, rotate = 135] [color={rgb, 255:red, 0; green, 0; blue, 0 }  ][line width=0.75]    (10.93,-3.29) .. controls (6.95,-1.4) and (3.31,-0.3) .. (0,0) .. controls (3.31,0.3) and (6.95,1.4) .. (10.93,3.29)   ;
\draw    (201.41,101.41) -- (250,150) ;

\draw [shift={(200,100)}, rotate = 45] [color={rgb, 255:red, 0; green, 0; blue, 0 }  ][line width=0.75]    (10.93,-3.29) .. controls (6.95,-1.4) and (3.31,-0.3) .. (0,0) .. controls (3.31,0.3) and (6.95,1.4) .. (10.93,3.29)   ;
\draw    (150,150) -- (198.59,198.59) ;
\draw [shift={(200,200)}, rotate = 225] [color={rgb, 255:red, 0; green, 0; blue, 0 }  ][line width=0.75]    (10.93,-3.29) .. controls (6.95,-1.4) and (3.31,-0.3) .. (0,0) .. controls (3.31,0.3) and (6.95,1.4) .. (10.93,3.29)   ;

\draw    (101.41,201.41) -- (150,250) ;

\draw [shift={(100,200)}, rotate = 45] [color={rgb, 255:red, 0; green, 0; blue, 0 }  ][line width=0.75]    (10.93,-3.29) .. controls (6.95,-1.4) and (3.31,-0.3) .. (0,0) .. controls (3.31,0.3) and (6.95,1.4) .. (10.93,3.29)   ;
\draw  [color={rgb, 255:red, 0; green, 0; blue, 0 }  ,draw opacity=1 ][fill={rgb, 255:red, 0; green, 0; blue, 0 }  ,fill opacity=1 ] (321.29,50) .. controls (321.29,48.78) and (322.28,47.79) .. (323.5,47.79) .. controls (324.72,47.79) and (325.71,48.78) .. (325.71,50) .. controls (325.71,51.22) and (324.72,52.21) .. (323.5,52.21) .. controls (322.28,52.21) and (321.29,51.22) .. (321.29,50) -- cycle ;
\draw  [color={rgb, 255:red, 0; green, 0; blue, 0 }  ,draw opacity=1 ][fill={rgb, 255:red, 0; green, 0; blue, 0 }  ,fill opacity=1 ] (197.79,100) .. controls (197.79,98.78) and (198.78,97.79) .. (200,97.79) .. controls (201.22,97.79) and (202.21,98.78) .. (202.21,100) .. controls (202.21,101.22) and (201.22,102.21) .. (200,102.21) .. controls (198.78,102.21) and (197.79,101.22) .. (197.79,100) -- cycle ;
\draw  [color={rgb, 255:red, 0; green, 0; blue, 0 }  ,draw opacity=1 ][fill={rgb, 255:red, 0; green, 0; blue, 0 }  ,fill opacity=1 ] (247.79,150) .. controls (247.79,148.78) and (248.78,147.79) .. (250,147.79) .. controls (251.22,147.79) and (252.21,148.78) .. (252.21,150) .. controls (252.21,151.22) and (251.22,152.21) .. (250,152.21) .. controls (248.78,152.21) and (247.79,151.22) .. (247.79,150) -- cycle ;
\draw  [color={rgb, 255:red, 0; green, 0; blue, 0 }  ,draw opacity=1 ][fill={rgb, 255:red, 0; green, 0; blue, 0 }  ,fill opacity=1 ] (147.79,150) .. controls (147.79,148.78) and (148.78,147.79) .. (150,147.79) .. controls (151.22,147.79) and (152.21,148.78) .. (152.21,150) .. controls (152.21,151.22) and (151.22,152.21) .. (150,152.21) .. controls (148.78,152.21) and (147.79,151.22) .. (147.79,150) -- cycle ;
\draw  [color={rgb, 255:red, 0; green, 0; blue, 0 }  ,draw opacity=1 ][fill={rgb, 255:red, 0; green, 0; blue, 0 }  ,fill opacity=1 ] (197.79,200) .. controls (197.79,198.78) and (198.78,197.79) .. (200,197.79) .. controls (201.22,197.79) and (202.21,198.78) .. (202.21,200) .. controls (202.21,201.22) and (201.22,202.21) .. (200,202.21) .. controls (198.78,202.21) and (197.79,201.22) .. (197.79,200) -- cycle ;
\draw  [color={rgb, 255:red, 0; green, 0; blue, 0 }  ,draw opacity=1 ][fill={rgb, 255:red, 0; green, 0; blue, 0 }  ,fill opacity=1 ] (147.79,250) .. controls (147.79,248.78) and (148.78,247.79) .. (150,247.79) .. controls (151.22,247.79) and (152.21,248.78) .. (152.21,250) .. controls (152.21,251.22) and (151.22,252.21) .. (150,252.21) .. controls (148.78,252.21) and (147.79,251.22) .. (147.79,250) -- cycle ;
\draw  [color={rgb, 255:red, 0; green, 0; blue, 0 }  ,draw opacity=1 ][fill={rgb, 255:red, 0; green, 0; blue, 0 }  ,fill opacity=1 ] (47.79,250) .. controls (47.79,248.78) and (48.78,247.79) .. (50,247.79) .. controls (51.22,247.79) and (52.21,248.78) .. (52.21,250) .. controls (52.21,251.22) and (51.22,252.21) .. (50,252.21) .. controls (48.78,252.21) and (47.79,251.22) .. (47.79,250) -- cycle ;
\draw  [color={rgb, 255:red, 0; green, 0; blue, 0 }  ,draw opacity=1 ][fill={rgb, 255:red, 0; green, 0; blue, 0 }  ,fill opacity=1 ] (97.79,200) .. controls (97.79,198.78) and (98.78,197.79) .. (100,197.79) .. controls (101.22,197.79) and (102.21,198.78) .. (102.21,200) .. controls (102.21,201.22) and (101.22,202.21) .. (100,202.21) .. controls (98.78,202.21) and (97.79,201.22) .. (97.79,200) -- cycle ;
\draw [color={rgb, 255:red, 155; green, 155; blue, 155 }  ,draw opacity=1 ][line width=1.5]  [dash pattern={on 1.69pt off 2.76pt}]  (239.33,238.33) -- (262.08,260.33) ;

\draw [color={rgb, 255:red, 155; green, 155; blue, 155 }  ,draw opacity=1 ][line width=1.5]  [dash pattern={on 1.69pt off 2.76pt}]  (289.67,188.33) -- (312.42,210.33) ;

\draw [color={rgb, 255:red, 155; green, 155; blue, 155 }  ,draw opacity=1 ][line width=1.5]  [dash pattern={on 1.69pt off 2.76pt}]  (411.21,239.25) -- (386.75,261.75) ;

\draw [color={rgb, 255:red, 155; green, 155; blue, 155 }  ,draw opacity=1 ][line width=1.5]  [dash pattern={on 1.69pt off 2.76pt}]  (363.33,187.33) -- (339.54,209.5) ;

\draw (322,34) node   {$c_{0}$};
\draw (500,100) node   {$c_{1}$};
\draw (100,150) node   {$c_{2}$};
\draw (50,200) node   {$c_{3}$};
\draw (20,250) node   {$c_{0}$};
\draw (300,150) node   {$c_{2}$};
\draw (250,200) node   {$c_{3}$};
\draw (183.83,250) node   {$c_{0}$};
\draw (550,150) node   {$c_{2}$};
\draw (600,200) node   {$c_{3}$};
\draw (631,250) node   {$c_{0}$};
\draw (365,150) node   {$c_{2}$};
\draw (400,200) node   {$c_{3}$};
\draw (450,250) node   {$c_{0}$};
\draw (150,100) node   {$c_{1}$};
\end{tikzpicture}
}
\end{center}
\caption{Shape of intended models in the proof of Theorem~\ref{t:expharda}. All worlds are reflexive.} 
\label{fig:trees}
\end{figure}

Now we show that $\varphi^{*}$ is globally satisfiable in $\MLKFour(\Diamond, \RevDiamond)$,~$\MLDFour(\Diamond, \RevDiamond)$ and~$\MLSFour(\Diamond, \RevDiamond)$.
First, note that due to the construction $R'$ is transitive and reflexive (and thus also serial).
Next, note that the second part of the formula $\varphi^*$ is globally satisfied in $\str{A}'$, since every world belongs to exactly one of $V'(c_i)$ (due to the fact that the satisfaction of $c_i$ depends on a distance from the root, which is unique since $\str{A}$ is assumed to be  tree-shaped).
Finally, we show that $\str{A}' \models \tr(\varphi)$. The proof is by induction, where the inductive hypothesis states that for any subformula $\psi$ of $\varphi$
and every world $w$ we have $\str{A}, w \models \psi$ if and only if $\str{A}', w \models \tr(\psi)$. 
The case of $\psi$ being a propositional variable follows from the second item of definition of $\str{A}'$. 
The case when $\psi$ is a Boolean combination of formulas is immediate from the inductive hypothesis and the semantics of~$\models$. 
Hence, the only interesting case is when $\psi$ is of the form $\Diamond(\psi')$.
We prove only one implication; the second one is analogous. Assume that $\str{A}, w \models \Diamond(\psi')$.
Thus there is a world $v$ such that $R(w,v)$ and $\str{A}, v \models \psi'$.
By induction hypothesis we deduce that $\str{A}', v \models \tr(\psi')$. 
Moreover, for $i = d(w) \; \mathop{mod} \; 4 $ and $j = d(v) \; \mathop{mod} \; 4 = (i+1) \; \mathop{mod} \; 4$
we have $\str{A}', w \models c_i$ and $\str{A}', v \models c_j$. Moreover, if $i$ is even then $(w,v) \in R'$ and $(v,w) \in R'$ otherwise.
In each of the cases~$i \in \{ 0, 1, 2, 3\}$  these all imply that $\str{A}', w \models \tr(\psi)$, which finishes the proof.
\end{proof}

\begin{cla}
If  $\varphi^*$ is globally satisfiable in $\MLKFour(\Diamond, \RevDiamond)$,~$\MLDFour(\Diamond, \RevDiamond)$ or~$\MLSFour(\Diamond, \RevDiamond)$ then $\varphi$ is globally satisfiable in $\MLK(\Diamond)$.
\end{cla}
\begin{proof}
Let $\str{A}= \langle W, R, V \rangle$ be a model of $\varphi^*$. 
We will define an increasing chain of structures $\str{A}'_0$, $\str{A}'_1$, $\ldots$, in which $\str{A}_i'=\langle W'_i, R'_i, V'_i \rangle$,  together with a 
pattern  function $f: A_0 \cup A_1 \cup \ldots \rightarrow W$. The  Kripke structure $\str{A}' = \langle W', R', V' \rangle$ defined as the union of the chain  
will turn out to be a model of $\varphi$. Our chain of structures is defined as follows.

 We fix a world $w \in W$, set $\str{A}_0' = \langle \{ w' \}, \emptyset, V_0' \rangle$ with $V_0'(w') = V(w)$ and set $f(w')=w$.
For simplicity let us assume that $\str{A}, w \models c_0$. In our construction,  
for every element $w'$ freshly added to $\str{A}'_i$, we will have that $f(w')$ satisfy $c_{i \mod 4}$.

Assume now that $\str{A}'_i$ is defined. To construct $\str{A}'_{i+1}$ we repeat for every element $w'$ freshly added to $\str{A}'_i$:
if $i$ is even (odd) then for every $R$-successor ($R$-predecessor) $v$ of $f(w')$ in $\str{A}$ such that $\str{A}, v \models c_{i+1 \mod 4}$ add to $W_{i+1}$ a fresh $R$-successor $v'$ of $w'$  and let $V'_i(v')=V(v)$ and $f(v')=v$.

We prove inductively over the shape of $\psi$ that $\str{A}', w' \models \psi$ iff $\str{A}, f(w') \models \tr(\psi)$.
The case of atomic propositions and Boolean combinations follows immediately from the definition.
The only interesting case is of $\psi = \Diamond(\psi')$.
Here we show only one case of one implication; the other cases are analogous.
Assume that $\str{A}, f(w') \models \tr(\psi)$ holds as well as~$f(w') \models c_0$.
Then there is an $R$-successor $v$ of $f(w')$ satisfying $\tr(\psi') \land c_1$.
Note that the $R'$-successors of $w'$ are copies of $R$--successors of $f(w')$ satisfying $c_1$, thus
there is a world~$v'$ being an $R'$-successor of $f(w')$ and satisfying $f(v') = v$.
Hence, from the inductive assumption, we infer $\str{A}', v' \models \psi'$, which implies $\str{A}', w' \models \psi$.
Analyzing analogously the other cases  we finish the inductive proof of the claim.
\end{proof}

The two proceeding claims show the correctness of the translation, allowing us to conclude Theorem~\ref{t:expharda}.
\end{proof}

We next handle the case of Euclidean frames.

\begin{theorem} \label{t:k5hardness}
The global and local satisfiability problem for \MLKFive{}$(\Diamond, \RevDiamond)$ and \MLDFive~$(\Diamond, \RevDiamond)$ is \ExpTime-hard.
\end{theorem}
\begin{proof}
We explicitly consider the global satisfiability problem, but due to Lemma~\ref{lem:locisglobforeuclogics}
 our proof applies also to local satisfiability.
  The proof goes as the proof of Theorem~\ref{t:expharda}. Our current
  intended models are similar to the intended models there (as on
  Fig.~\ref{fig:trees}). The difference is that all the worlds
  satisfying~$c_1$ or~$c_3$ are made  equivalent to each other, and the worlds satisfying~$c_0$ or~$c_2$ are irreflexive. 
	Note that this does not violate the property that each world can identify its
  children in the tree. Observe also that such intended models are indeed
  Euclidean and serial (however, they are neither transitive nor reflexive); in particular all
  worlds satisfying~$c_0$ or~$c_2$ are lanterns.  Now, for a given $\MLK$ formula we can
  construct precisely the same formula as in the previous
  proof. We leave the routine details to the reader. The correctness proof
	is essentially identical to the correctness proof of Theorem~\ref{t:expharda} so we omit it here.
\end{proof}



\section{Conclusions}
We have filled the gaps remaining in the classification of the complexity of the local and global 
satisfiability problems for natural modal languages with graded and converse modalities over traditional classes of 
frames. What we have not systematically studied are 
the problem of combined satisfiability (given two formulas check if there exists a model in which the first is satisfied locally and the second is satisfied globally) and the problem of finite (local, global, combined) satisfiability (asking about the existence of \emph{finite} models). We suspect that 
the classification could be extended to cover these problems using results/techniques from our paper and the referenced articles without major obstacles.

Two other questions we leave open are if  the \NExpTime-lower bound in Thm.~\ref{thm:theorem5} remains valid if the numbers in graded modalities are encoded
in unary rather than in binary and if our finite model construction from Section \ref{s:transitive} can be replaced by a one producing smaller models.

\section*{Acknowledgements}
We thank Evgeny Zolin for providing us a comprehensive list of gaps in
the classification of the complexity of graded modal logics and for
sharing with us his tikz files with modal cubes.  We thank Emil
Je\v{r}\'abek for his explanations concerning
\MLKFive$(\Diamond, \RevDiamond)$.  We also thank Tomasz Gogacz and
Filip Murlak for comments concerning
Section~\ref{s:transitive}. Finally, we thank the anonymous reviewers
for their useful comments and remarks.

Bartosz Bednarczyk~is supported by Polish Ministry of Science 
and Higher Education program "Diamentowy Grant" no. DI2017 006447. 

Emanuel Kiero\'nski and Piotr Witkowski are supported by 
Polish National Science Centre grant no. 2016/21/B/ST6/01444.


\bibliographystyle{acmtrans}
\bibliography{bibliography}
\label{lastpage}

\end{document}